\documentclass[a4paper,11pt]{article}
\usepackage[margin = 1in]{geometry}
\usepackage[utf8]{inputenc}
\usepackage[T1]{fontenc}
\usepackage[dvipsnames]{xcolor}
\usepackage{graphicx}
\usepackage{amsmath, amssymb, amstext, mathtools}
\usepackage{amsthm}
\usepackage{tikz}
\usepackage{pgfplots}
\usepackage{float}
\usepackage{hyperref}
\usepackage{paralist}
\usepackage{enumitem}
\usepackage[titlenumbered, linesnumbered, ruled]{algorithm2e}

\allowdisplaybreaks

\definecolor[named]{urlblue}{cmyk}{1,0.58,0,0.21}

\hypersetup{
 breaklinks=true,
 colorlinks=true,
 citecolor=purple!70!blue!60!black,
 linkcolor=purple!70!blue!90!black,
 urlcolor=urlblue,
 pdflang={en},
 pdftitle={Bounding the Weisfeiler--Leman Dimension via a Depth Analysis of I/R-Trees},
 pdfauthor={Sandra Kiefer and Daniel Neuen}
}

\pgfplotsset{compat=1.18}
\usetikzlibrary{patterns,arrows,decorations.pathreplacing,decorations.pathmorphing,calc}

\tikzstyle{vertex}=[draw,circle,fill=white,minimum size=5pt,inner sep=0pt]
\tikzstyle{edge}=[draw,thick]

\definecolor{darkpastelgreen}{rgb}{0.01, 0.75, 0.24}

\newtheorem{theorem}{Theorem}[section]
\newtheorem{lemma}[theorem]{Lemma}
\newtheorem{corollary}[theorem]{Corollary}
\newtheorem{example}[theorem]{Example}

\theoremstyle{definition}
\newtheorem{definition}[theorem]{Definition}

\theoremstyle{remark}

\newtheorem{claim}[theorem]{Claim}

\newenvironment{claimproof}[1][\proofname]{
  
  \begin{proof}[#1]%
}{%
  \end{proof}%
}

\usepackage{eqparbox}
\newcommand\alignsymbols[2][c]{\mathrel{\eqmakebox[S][#1]{$#2$}}{}} 

\DeclareMathOperator{\depth}{depth}
\newcommand{\WL}[2]{{\sf WL}^{#1}\![#2]}
\newcommand{\WLit}[3]{{\sf WL}_{#2}^{#1}\![#3]}

\newcommand{\WLd}[2]{{\depth}_{\sf WL}^{#1}(#2)}

\newcommand{\logic}[1]{{\sf #1}}
\newcommand{\FO}{\logic{FO}}
\newcommand{\LC}{\logic C}
\newcommand{\LL}{\logic L}
\newcommand{\LCk}[1]{\LC^{#1}}

\newcommand{\ZZ}{\mathbb Z}

\newcommand{\CE}{\mathcal E}
\newcommand{\CH}{\mathcal H}
\newcommand{\CM}{\mathcal M}

\DeclareMathOperator{\im}{im}

\DeclareMathOperator{\tw}{tw}
\DeclareMathOperator{\dist}{dist}

\DeclareMathOperator{\flipf}{flip}
\DeclareMathOperator{\flip}{flip^*}

\DeclareMathOperator{\CFI}{CFI}
\DeclareMathOperator{\twCFI}{{CFI^{\textsf{x}}}}

\newcommand{\orcid}[1]{\href{https://orcid.org/#1}{\includegraphics[height=1.8ex]{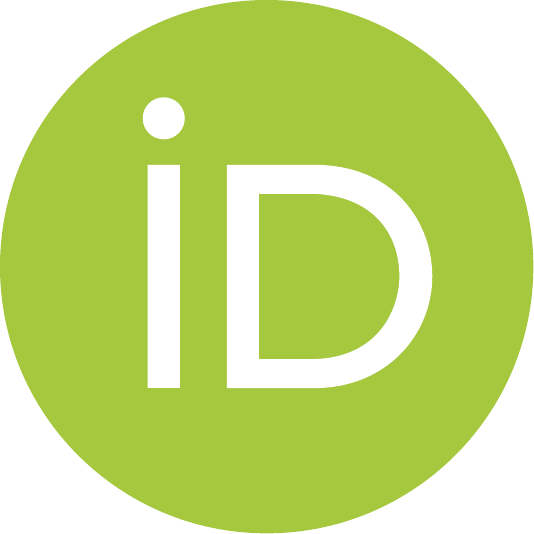}}}
\newcommand{\email}[1]{\href{mailto:#1}{\texttt{#1}}}

\title{Bounding the Weisfeiler--Leman Dimension\\via a Depth Analysis of I/R-Trees}
\author{
Sandra Kiefer\footnote{The research was supported by the Glasstone Benefaction, University of Oxford [Violette and Samuel Glasstone Research Fellowships in Science 2022] as well as Jesus College in Oxford.}\; \orcid{0000-0003-4614-9444}\\
University of Oxford\\
\email{sandra.kiefer@cs.ox.ac.uk}
\and
Daniel Neuen \orcid{0000-0002-4940-0318}\\
University of Bremen\\
\email{dneuen@uni-bremen.de}
}
\date{}

\begin{document}

\maketitle

\begin{abstract}
 The Weisfeiler--Leman (WL) dimension is an established measure for the inherent descriptive complexity of graphs and relational structures.
 It corresponds to the number of variables that are needed and sufficient to define the object of interest in a counting version of first-order logic (FO).
 These bounded-variable counting logics were even candidates to capture graph isomorphism, until a celebrated construction due to Cai, Fürer, and Immerman [Combinatorica 1992] showed that $\Omega(n)$ variables are required to distinguish all non-isomorphic $n$-vertex graphs.

 Still, very little is known about the precise number of variables required and sufficient to define every $n$-vertex graph.
 For the bounded-variable (non-counting) FO fragments, Pikhurko, Veith, and Verbitsky [Discret. Appl. Math. 2006] provided an upper bound of $\frac{n+3}{2}$ and showed that it is essentially tight.
 Our main result yields that, in the presence of counting quantifiers, $\frac{n}{4} + o(n)$ variables suffice.
 This shows that counting does allow us to save variables when defining graphs.
 As an application of our techniques, we also show new bounds in terms of the vertex cover number of the graph.

 To obtain the results, we introduce a new concept called the \emph{WL depth} of a graph. We use it to analyze branching trees within the Individualization/Refinement (I/R) paradigm from the domain of isomorphism algorithms.
 We extend the recursive procedure from the I/R paradigm by the possibility of splitting the graphs into independent parts.
 Then we bound the depth of the obtained branching trees, which translates into bounds on the WL dimension and thereby on the number of variables that suffice to define the graphs.
\end{abstract}

\section{Introduction}

The Weisfeiler--Leman (WL) algorithm is a combinatorial algorithm that iteratively collects and aggregates local information of structures and encodes them in colors assigned to parts of the structures.
The original algorithm introduced by Weisfeiler and Leman \cite{WeisfeilerL68} is the $2$-dimensional version, which colors pairs of vertices.
Its generalization to arbitrary dimension $k \geq 1$ is the $k$-dimensional WL algorithm ($k$-WL), which was independently introduced by Babai and Mathon (see \cite{Babai16} for a historic note) as well as by Immerman and Lander \cite{ImmermanL90}.

The most prominent application of the WL algorithm lies in the context of graph comparison, since the computed information can be used to discover structural differences in the graphs.
Indeed, the algorithm is commonly used as a subroutine in isomorphism algorithms both in practice (see, e.g., \cite{JunttilaK07,JunttilaK11,McKay81,McKayP14}) and theory (see, e.g., \cite{ArvindNPZ22,Babai80,BabaiCSTW13,Neuen22,Neuen24}), including Babai's quasipolynomial-time isomorphism test \cite{Babai16}, which falls back on $k$-WL for dimension $k = O(\log n)$.

Besides that, the WL algorithm has connections to numerous other areas in theoretical computer science (see the survey article \cite{Kiefer20}).
Among the most prominent links is the one to counting logics.
It is known that $k$-WL distinguishes two graphs if and only if the graphs can be distinguished via a sentence in first-order logic with counting quantifiers and $k+1$ variables, i.e., in the logic $\LCk{k+1}$ \cite{CaiFI92,ImmermanL90}.
Via this connection, the algorithm and the study of its expressive power have become a major tool to analyze the inherent descriptive complexity of relational structures.
More precisely, the \emph{WL dimension} of a graph $G$ is the minimum dimension of the algorithm that suffices to distinguish $G$ from every non-isomorphic graph and it directly corresponds to the minimum number of variables that suffice to define the graph in the counting logic $\LC$.

For some time, there was hope that there would be a universal bound on the WL dimension for all graphs, which would have placed the graph isomorphism problem in the complexity class P.
However, a celebrated construction due to Cai, Fürer, and Immerman \cite{CaiFI92} shows that $\Omega(n)$ variables are needed to define every graph on $n$ vertices, i.e., the WL dimension of $n$-vertex graphs is in $\Omega(n)$.
On the other hand, many restricted classes of graphs have finite WL dimension, for example, every graph class that excludes a fixed graph as a minor \cite{Grohe17}.

In recent years, there has been a series of results aiming to provide more precise bounds on the WL dimension of certain graph classes.
For example, Kiefer, Ponomarenko and Schweitzer \cite{KieferPS19} show that the WL dimension of planar graphs is at most $3$.
Further examples include linear bounds (with small and explicit constant factors) on the WL dimension of graphs of bounded tree-width \cite{KieferN22}, rank-width \cite{GroheN23} and genus \cite{GroheK19}.

Surprisingly, there has only been little progress towards determining the WL dimension only in terms of the number of vertices of a graph, which is arguably the simplest and most natural graph parameter.
In terms of logics, this poses the following question.
\textit{How many variables are needed and sufficient to define graphs on $n$ vertices in $\LC{}$?}
Pikhurko, Veith, and Verbitsky \cite{PikhurkoVV06} proved that for every two non-isomorphic $n$-vertex graphs $G$ and $H$, there is an $\FO$-formula with at most $\frac{n+3}{2}$ variables that distinguishes $G$ and $H$.
For the logic $\FO$ (without counting quantifiers), this bound is optimal up to an additive term of one, i.e., at least $\frac{n+1}{2}$ variables are required to distinguish non-isomorphic $n$-vertex graphs in $\FO$.
Since every $\FO$-formula is also a $\LC{}$-formula, the upper bound also holds for $\LC{}$, i.e., every $n$-vertex graph $G$ can be defined in $\LC{}$ with at most $\frac{n+3}{2}$ variables.\footnote{Note that, while we can well speak about definability of a graph in $\LC{}$, i.e., it being distinguished from every second graph, we need to restrict ourselves to the comparison of graphs with equal numbers of vertices in the bounded-variable fragments of $\FO$, since we cannot determine the number $n$ of vertices without using $n+1$ variables or counting.}
However, as opposed to the $\FO$ setting, it has remained open whether this bound is tight in the presence of counting quantifiers or whether they actually decrease the number of distinct variables that are needed.

Taking a closer look at the Cai-Fürer-Immerman (CFI) construction, \cite{PikhurkoVV06} also shows that at least $0.00465n$ variables are required to define every $n$-vertex graph $G$ in $\LC{}$.
Using today's refined understanding of the CFI construction (see, e.g., \cite{DawarR07, Neuen23, Roberson22}) and more recent results on expander graphs \cite{DvorakN16}, we can obtain an improved lower bound of $(\frac{1}{96} - o(1))n$ on the WL dimension of all $n$-vertex graphs.
Still, there is a considerable gap between this lower and the upper bound.

\paragraph{This work.}

Our main contribution is the following upper bound on the WL dimension of the class of all $n$-vertex graphs, which improves on the result obtained by Pikhurko, Veith, and Verbitsky \cite{PikhurkoVV06}.

\begin{theorem}
 \label{thm:main}
 The WL dimension of every $n$-vertex graph is at most $\frac{n}{4} + o(n)$.
\end{theorem}

Via the characterization of the WL algorithm in terms of first-order logic with counting quantifiers, we immediately obtain the following.

\begin{corollary}
 \label{cor:main}
 Every $n$-vertex graph $G$ can be defined in $\LCk{k}$, i.e., in the $k$-variable fragment of first-order logic enriched with counting quantifiers, where $k \in \frac{n}{4} + o(n)$.
\end{corollary}

In particular, this shows that the ability to count does allow us to save variables when defining graphs.

Our techniques also allow us to obtain an improved bound on the WL dimension in terms of the vertex cover number.
For a graph $G$, a set $S \subseteq V(G)$ is a \emph{vertex cover} of $G$ if $e \cap S \neq \emptyset$ for all edges $e \in E(G)$, and the \emph{vertex cover number} of $G$ is the minimal size of a vertex cover.
It follows from \cite{KieferN22} that the WL dimension of a graph is at most its vertex cover number.
We obtain the following stronger bound.

\begin{theorem}
 \label{thm:main-vc}
 The WL dimension of every graph with vertex cover number $r$ is at most $\frac{2}{3}r + 3$.
\end{theorem}

The starting point for our analysis is the following observation.
Suppose that for a graph $G$, there are vertices $v_1,\dots,v_s$ such that, after individualizing these vertices (i.e., assigning a unique color to each of them) and performing $1$-WL (also known as the \emph{Color Refinement algorithm}), we obtain a discrete coloring, i.e.\ one where all vertices receive different colors.
Then it is well known that the WL dimension of $G$ is at most $s+1$.
This idea builds the foundation of the Individualization/Refinement (I/R) paradigm, which is most prominently used in practical graph isomorphism tools, see, e.g., \cite{JunttilaK07,JunttilaK11,McKay81,McKayP14}, but it has also been used in theoretical work, such as \cite{Babai80, Babai81, BabaiCSTW13, KieferPS19, Spielman96}.

Unfortunately, the method does not suffice to obtain good upper bounds on the WL dimension of $n$-vertex graphs, since the number of vertices that need to be individualized might be high.
Indeed, for the complete graph $K_n$ on $n$ vertices, it is easy to see that $n-1$ vertices need to be individualized so that applying $1$-WL results in a discrete coloring.
To overcome this issue, we adopt a strategy inspired by practical isomorphism solvers \cite{JunttilaK11}, which is to allow for simple modifications to the graph without changing the problem at hand, as well as treating components of a graph independently.
For example, the WL dimension of a graph $G$ is equal to the WL dimension of $G$ with complemented edge set. For $K_n$, complementing the edge set yields isolated vertices.
Now, using that the WL dimension of a graph $G$ is at most $\max\{2,k\}$, where $k$ is the maximum dimension of its connected components, we obtain that the WL dimension of a complete graph is at most $2$.

It turns out that a combination of these ideas is already very powerful and suffices for us to show Theorems \ref{thm:main} and \ref{thm:main-vc}. We concretize the approach in a new concept, which we call the \emph{WL depth} of a graph.  
To obtain our results, we first show that if $G$ has $k$-WL depth at most $\ell$, then its WL dimension is at most $\max\{2,k\} + \ell$.
We then prove that the $1$-WL depth of a graph is at most $\frac{2}{3}r + 1$, where $r$ is the vertex cover number of the graph.
This implies Theorem \ref{thm:main-vc}.
Afterwards, we show that the $2$-WL depth of any $n$-vertex graph is at most $\frac{n}{4} + o(n)$, resulting in Theorem \ref{thm:main}.
This part is our main technical contribution and may be of independent interest since it also provides bounds on the possible depth of branching trees considered by practical isomorphism tests \cite{JunttilaK11}.

\paragraph{Outline.}

After covering the necessary preliminaries in the next section, we introduce the WL depth of a graph in Section \ref{sec:wl-depth} and prove various basic properties for it.
This notion is the key concept underlying our main results.
After that, we prove Theorem \ref{thm:main-vc} in Section \ref{sec:vc} and Theorem \ref{thm:main} in Section \ref{sec:advanced}.
Finally, we prove a more precise lower bound on the WL dimension of $n$-vertex graphs in Section \ref{sec:lower-bound}.

\section{Preliminaries}

\subsection{Graphs and Colorings}

An (undirected) \emph{graph} is a pair $G = (V(G),E(G))$ of a finite, non-empty \emph{vertex set} $V(G)$ and an \emph{edge set} $E(G) \subseteq \big\{\{u,v\} \, \big\vert \, u,v \in V(G), u \neq v\big\}$.
For $v,w \in V(G)$, we also write $vw$ as a shorthand for $\{v,w\}$.
The \emph{neighborhood} of~$v$ in $G$ is~$N_G(v) \coloneqq \{w \mid \{v,w\} \in E(G)\}$ and the \emph{degree} of $v$ in $G$ is $\deg_G(v) \coloneqq |N_G(v)|$.
The \emph{closed neighborhood} of $v$ in $G$ is $N_G[v] \coloneqq N_G(v) \cup \{v\}$.
For $W \subseteq V(G)$, we also define $N_G(W) \coloneqq \left(\bigcup_{v \in W}N(v)\right) \setminus W$.
If the graph is clear from the context, we usually omit the subscript.

For a graph $G$ and $U,W \subseteq V(G)$, we define $E_G(U,W) \coloneqq \{uw \in E(G) \mid u \in U, w \in W\}$.
We also denote by $G[W]$ the \emph{induced subgraph} of $G$ on the vertex set $W$, and define $G - W \coloneqq G[V(G) \setminus W]$.
For disjoint subsets $U,W \subseteq V(G)$, we define the bipartite graph $G[U,W]$ with vertex set $U \cup W$ and edge set $E_G(U,W)$.
A bipartite graph $G$ with bipartition $(V_1,V_2)$ is called \emph{biregular} if $\deg(v_1) = \deg(v_1')$ for all $v_1,v_1' \in V_1$, and $\deg(v_2) = \deg(v_2')$ for all $v_2,v_2' \in V_2$.

For a graph $G$ and two vertices $v,w \in V(G)$, a \emph{path of length $\ell$} from $v$ to $w$ is a sequence $v = u_0,u_1,\dots,u_\ell = w$ of pairwise distinct vertices such that $u_{i-1}u_i \in E(G)$ for all $i \in [\ell]$ (where $[\ell] \coloneqq \{1,\dots,\ell\}$).
The \emph{distance} between $v$ and $w$, denoted by $\dist_G(v,w)$, is the minimal length of a path from $v$ to $w$.

A set $S \subseteq V(G)$ is a \emph{vertex cover} of $G$ if $S \cap e \neq \emptyset$ for all $e \in E(G)$.
The \emph{vertex cover number} of $G$ is the smallest integer $r \geq 0$ such that $G$ has a vertex cover of size $r$.

A \emph{(vertex-)colored graph} is a tuple $(G,\chi)$ where $G$ is a graph and $\chi\colon V(G) \rightarrow C$ is a \emph{(vertex) coloring}, a mapping from $V(G)$ into some set $C$ of colors.
We write $\im(\chi) \coloneqq \{\chi(v) \mid v \in V(G)\}$ to denote the image of $\chi$. We generally assume that all graphs are colored even if not explicitly stated.
Typically, the color set $C$ is chosen to be an initial segment $[n]$ of the natural numbers.
We say a coloring $\chi$ is \emph{discrete} if it is injective, i.e., all color classes have size~$1$.

Let $\chi\colon U \rightarrow C$ be a coloring of the elements of some universe $U$.
For a tuple $\bar u = (u_1,\dots,u_\ell) \in U^\ell$, we define $\chi[\bar u]$ to be the coloring obtained from $\chi$ by individualizing all elements $u_1,\dots,u_\ell$, i.e., $(\chi[\bar u])(u_i) = (1,\min\{j \in [\ell] \mid u_i = u_j\})$ for all $i \in [\ell]$, and $(\chi[\bar u])(v) = (0,\chi(v))$ for all $v \in U \setminus \{u_1,\dots,u_\ell\}$.
For $u_1,\dots,u_\ell \in U$, we also write $\chi[u_1,\dots,u_\ell]$ for $\chi[(u_1,\dots,u_\ell)]$.
Moreover, slightly abusing notation, for a set $S = \{u_1,\dots,u_\ell\} \subseteq U$, we write $\chi[S]$ instead of $\chi[u_1,\dots,u_\ell]$ if the order of the elements is irrelevant to us.

An \emph{isomorphism} from $G$ to a graph $H$ is a bijection $\varphi\colon V(G) \rightarrow V(H)$ such that for all~$v,w \in V(G)$, it holds that~$\{v,w\} \in E(G)$ if and only if $\{\varphi(v),\varphi(w)\} \in E(H)$.
The graphs $G$ and $H$ are \emph{isomorphic} ($G \cong H$) if there is an isomorphism from~$G$ to~$H$.
We write $\varphi\colon G\cong H$ to denote that $\varphi$ is an isomorphism from $G$ to $H$.
Isomorphisms between colored graphs have to respect the colors of the vertices.

\subsection{Logics}

We give a brief introduction to the logics and notation we use here. For more background, we refer the reader to \cite{Grohe17,Otto17}. We write $\FO$ to denote standard \emph{first-order logic}.
For a formula $\varphi$, we write $\varphi(x_1,\dots,x_k)$ to indicate that the free variables of $\varphi$ are among the variables $\{x_1,\dots,x_k\}$.
A \emph{sentence} is a formula without free variables.
Let $G$ be a graph and $v_1,\dots,v_k \in V(G)$.
We write $G \models \varphi(v_1,\dots,v_k)$ if $G$ is a model of $\varphi$ when $x_i$ is interpreted by $v_i$.

We define \emph{first-order logic with counting quantifiers} $\LC$ to be the extension of $\FO$ by counting quantifiers of the form $\exists^{\geq j} x \varphi$.
The formula $\exists^{\geq j} x \varphi$ is satisfied over a graph $G$ if there are at least $j$ distinct elements $v \in V(G)$ that satisfy $\varphi$.
For $k \geq 1$, we define $\LCk{k}$ to be the restriction of $\LC$ to formulas with at most $k$ variables, i.e., we restrict ourselves to a set of variables of size $k$.

Let $G$ and $H$ be two graphs and let $\LL$ be a logic.
We say that a sentence $\varphi \in \LL{}$ \emph{distinguishes} $G$ and $H$ if $G \models \varphi \iff H \not\models \varphi$.

\subsection{The Weisfeiler--Leman Algorithm}

Let~$\chi_1,\chi_2\colon (V(G))^k \rightarrow C$ be colorings of the~$k$-tuples of vertices of a graph~$G$. 
We say $\chi_1$ \emph{refines} $\chi_2$, denoted $\chi_1 \preceq \chi_2$, if $\chi_1(\bar v) = \chi_1(\bar w)$ implies $\chi_2(\bar v) = \chi_2(\bar w)$ for all $\bar v,\bar w \in (V(G))^{k}$.
The colorings $\chi_1$ and $\chi_2$ are \emph{equivalent}, denoted $\chi_1 \equiv \chi_2$,  if $\chi_1 \preceq \chi_2$ and $\chi_2 \preceq \chi_1$.
The coloring $\chi_1$ \emph{strictly refines} $\chi_2$, denoted $\chi_1 \prec \chi_2$, if $\chi_1 \preceq \chi_2$ and $\chi_1 \not\equiv \chi_2$.

We describe the \emph{$k$-dimensional Weisfeiler--Leman algorithm} ($k$-WL) for $k \geq 1$.
Let $(G,\chi)$ be a colored graph.
We define the initial coloring $\WLit{k}{0}{G,\chi}\colon (V(G))^{k} \rightarrow C$ as the coloring where each tuple is colored with the isomorphism type of its underlying ordered subgraph.
More precisely, for two colored graphs $(G,\chi)$ and $(G,\chi')$, and vertices $v_1,\dots,v_k \in V(G)$ and $v_1',\dots,v_k' \in V(G')$ we have $\WLit{k}{0}{G,\chi}(v_1,\dots,v_k) = \WLit{k}{0}{G',\chi'}(v_1',\dots,v_k')$
if and only if there is an isomorphism between the induced colored graphs $G[\{v_1,\dots,v_k\}]$ and $G'[\{v'_1,\dots,v'_k\}]$ mapping $v_i$ to $v'_i$ for all $i \in [k]$.

We then recursively define the coloring $\WLit{k}{i}{G,\chi}$ obtained after $i$ rounds of the algorithm.
Suppose $k \geq 2$.
For $\bar v = (v_1,\dots,v_k) \in (V(G))^k$, set
$\WLit{k}{i+1}{G,\chi}(\bar v) \coloneqq \big(\WLit{k}{i}{G,\chi}(\bar v), \CM_i(\bar v)\big)$,
where
\[\CM_i(\bar v) \mspace{-3mu}\coloneqq\mspace{-3mu} \Big\{\!\!\Big\{\mspace{-2mu}\big(\WLit{k}{i}{G,\mspace{-3mu}\chi}(\bar v[\mspace{-1mu}w\mspace{-1mu}/\mspace{-1mu}1\mspace{-1mu}]),\mspace{-0.5mu}...\mspace{1.2mu},\mspace{-1.5mu}\WLit{k}{i}{G,\mspace{-3mu}\chi}(\bar v[\mspace{-1mu}w\mspace{-1mu}/\mspace{-1mu}k\mspace{-1mu}])\mspace{-1mu}\big) \;\mspace{-2mu}\Big\vert\mspace{-2mu}\; w \mspace{-2mu}\in\mspace{-2mu} V\mspace{-1mu}(G)\mspace{-2mu} \Big\}\!\!\Big\}\]
and $\bar v[w/i] \coloneqq (v_1,\dots,v_{i-1},w,v_{i+1},\dots,v_k)$ is the tuple obtained from substituting the $i$-th entry of $\bar v$ with $w$.

For $k = 1$, the definition is analogous, but we set
\[\CM_i(\bar v) \coloneqq \Big\{\!\!\Big\{ \WLit{1}{i}{G,\chi}(w) \;\Big\vert\; w \in N_G(v) \Big\}\!\!\Big\}.\]

By definition, $\WLit{k}{i+1}{G,\chi} \preceq \WLit{k}{i}{G,\chi}$ holds for all $i \geq 0$.
So there is a minimal~$i_\infty \geq 0$ such that $\WLit{k}{i_{\infty}}{G,\chi} \equiv \WLit{k}{i_{\infty}+1}{G,\chi}$, and we set $\WL{k}{G,\chi} \coloneqq \WLit{k}{i_\infty+1}{G,\chi}$.

To simplify notation, for $\ell < k$ and vertices $v_1,\dots,v_\ell \in V(G)$, we write $\WL{k}{G,\chi}(v_1,\dots,v_\ell)$ for $\WL{k}{G,\chi}(v_1,\dots,v_\ell,v_\ell,\dots,v_\ell)$, where the latter tuple has $k$ entries.

Also, we say a vertex coloring $\lambda\colon V(G) \to C$ is \emph{stable with respect to $k$-WL} if it is not refined by $k$-WL.
Formally, this means that $\lambda \equiv \lambda'$ where $\lambda'\colon V(G) \to C$ is defined via $\lambda'(v) \coloneqq \WL{k}{G,\lambda}(v)$ for all $v \in V(G)$.

The algorithm $k$-WL takes a vertex-colored graph $(G,\chi)$ as input and returns $\WL{k}{G,\chi}$.
Given vertex-colored graphs $(G,\chi)$ and $(H,\lambda)$, the algorithm \emph{distinguishes} $(G,\chi)$ and $(H,\lambda)$ if
\[\Big\{\!\!\Big\{\WL{k}{G,\chi}(\bar v) \;\Big\vert\; \bar v \mspace{-1.5mu}\in\mspace{-1.5mu} (V(G))^k\Big\}\!\!\Big\} \mspace{-1mu} \neq \mspace{-1mu} \Big\{\!\!\Big\{\WL{k}{H,\lambda}(\bar w) \;\Big\vert\; \bar w \mspace{-1.5mu}\in\mspace{-1.5mu} (V(H))^k\Big\}\!\!\Big\}.\]
We write $(G,\chi) \simeq_k (H,\lambda)$ if $k$-WL does not distinguish $(G,\chi)$ and $(H,\lambda)$.
Also, $k$-WL \emph{identifies} $(G,\chi)$ if it distinguishes $(G,\chi)$ from every other non-isomorphic vertex-colored graph. From the definition, it follows that all graphs that are distinguished (and identified, respectively) by $k$-WL are also distinguished (and identified, respectively) by $(k+1)$-WL.
The \emph{WL dimension} of a graph $(G,\chi)$ is the minimal integer $k \geq 1$ such that $k$-WL identifies $(G,\chi)$.

The following theorem connects the WL algorithm to first-order logic with counting quantifiers.

\begin{theorem}[\cite{CaiFI92,ImmermanL90}]
 \label{thm:eq-wl-ck}
 Let $k \geq 1$ and let $G$, $H$ be two graphs.
 Then $G \simeq_k H$ if and only if there is no sentence in $\LCk{k}$ that distinguishes $G$ and $H$.
\end{theorem}

We also require the following lemma.

\begin{lemma}
 \label{la:indiv-dimension}
 Let $k \in \mathbb{N}$, let $(G,\chi)$ be a vertex-colored graph, and let $u \in V(G)$.
 Suppose $k$-WL identifies the graph $(G, \chi[u])$.
 Then $(k + 1)$-WL identifies $(G,\chi)$.
\end{lemma}

\begin{proof}
 Suppose there is a vertex-colored graph $(G',\chi')$ with a vertex $u' \in V(G')$ such that $\WL{k+1}{G',\chi'}(u') = \WL{k+1}{G,\chi}(u)$.
 Hence, we have that
 \begin{align*}
     &\{\!\!\{\WL{k+1}{G,\chi}(u, w_1, \dots, w_k) \mid (w_1, \dots , w_k) \in V^k(G)\}\!\!\}\\
  =~ &\{\!\!\{\WL{k+1}{G',\chi'}(u', w'_1, \dots, w'_k) \mid (w'_1, \dots , w'_k) \in V^k(G')\}\!\!\}.
 \end{align*}
 Thus, the graphs $(G,\chi[u])$ and $(G',\chi'[u'])$ obtain equal colorings under $k$-WL.
 By assumption, this implies that $(G,\chi[u])$ and $(G',\chi'[u'])$ are isomorphic, which is equivalent to the existence of an isomorphism from $(G,\chi)$ to $(G',\chi')$ mapping $u$ to $u'$.
\end{proof}

\section{The Weisfeiler--Leman Depth}
\label{sec:wl-depth}

In this section, we introduce the $k$-WL depth of a colored graph $(G,\chi)$, which is the key notion underlying our analysis.
Let $k \geq 1$.
Intuitively speaking, the $k$-WL depth is the minimum number of vertices that need to be individualized in order to obtain a discrete coloring after performing $k$-WL, except that we allow for some simple additional operations that do not affect the WL dimension.
The first such operation is the possibility to split a graph into connected components.
Indeed, it is known that if $k \geq 2$ and $k$-WL identifies every connected component of a graph, then $k$-WL also identifies $G$.
So we may always restrict ourselves to the analysis of connected graphs when we are interested in bounding the WL dimension.
In itself, this operation is of limited use, since it can only be applied at most once.
For this reason, we also consider operations that allow us to decrease the number of edges (hoping to make the graph disconnected again).

The most basic operation is to complement the entire edge set of $G$, which is already sufficient to handle complete graphs.
However, in a vertex-colored graph, we can also complement edges between two color classes without changing the WL dimension.
More generally, we can choose any number of pairs of vertex colors and complement edges exactly between the chosen pairs.
To formalize this idea, we introduce the notion of a flip function.

Let $(G,\chi)$ be a vertex-colored graph and let $C \coloneqq \im(\chi)$.
A \emph{flip function for $(G,\chi)$} is a function $f\colon C \times C \to \{0,1\}$ such that $f(c_1,c_2) = f(c_2,c_1)$ for all $c_1,c_2 \in C$.

Let $f$ be a flip function for $(G,\chi)$.
We define the \emph{$f$-flip of $(G,\chi)$} to be the graph $\flipf_f(G,\chi) \coloneqq (G',\chi)$ where $V(G') = V(G)$ and
\begin{align*}
 E(G') \coloneqq\;\;\; &\{vw \in E(G) \mid f(\chi(v),\chi(w)) = 0\}\\
            \cup\;     &\{vw \notin E(G) \mid v \neq w, f(\chi(v),\chi(w)) = 1\}.
\end{align*}
Intuitively speaking, the $f$-flip of $(G,\chi)$ is obtained from $(G,\chi)$ by complementing the edges between all color classes associated with $c_1,c_2 \in C$ for which $f(c_1,c_2) = 1$.

We will mostly use one particular flip function $f$, which aims at minimizing the number of edges in the $f$-flip.
We define $f^*_{G,\chi} \colon C \times C \to \{0,1\}$ via $f^*_{G,\chi}(c_1,c_2) = 1$ if and only if
\[\left|E_G(\chi^{-1}\mspace{-1mu}(c_1),\mspace{-2mu}\chi^{-1}\mspace{-2mu}(c_2))\right| > \mspace{-2mu}\frac{1}{2}\left|\Big\{vw \mspace{-2mu}\in\mspace{-2mu} \binom{V(G)}{2} \;\Big|\; \chi(v) \mspace{-1mu}= \mspace{-1mu}c_1, \mspace{-2mu}\chi(w) \mspace{-1mu}= \mspace{-1mu}c_2\Big\}\right|\mspace{-1mu}.\]
We define the \emph{flip} of $(G,\chi)$ as $\flip(G,\chi) \coloneqq \flipf_{f_{G,\chi}^*}(G,\chi)$.
We say that $(G,\chi)$ \emph{is flipped} if $(G,\chi) = \flip(G,\chi)$.
Observe that flipping is idempotent, i.e., the flip of $(G,\chi)$ is flipped.

With these definitions at hand, we are now ready to formally define the $k$-WL depth.
As already indicated above, the concept aims at describing the minimum number of vertices that need to be individualized in order to obtain a discrete coloring after performing $k$-WL, but we additionally allow splitting the graph into its connected components (which are then treated independently) and transitioning to the $f$-flip for any flip function $f$.
More precisely, we wish to arrive at a graph with a discrete coloring by applying the operations (1) refine the coloring using $k$-WL, (2) move to the $f$-flip for any flip function $f$, (3) split the graph into connected components and treat them independently, and (4) individualize some vertex.
In doing so, we aim to minimize the number of vertex individualizations.
This is formalized in the next definition.

\begin{definition}
 \label{def:irc}
 Let $k \geq 1$.
 A \emph{$k$-IRC tree} of $(G,\chi)$ is a triple $(T,r,\gamma)$ where $T$ is a tree with root $r \in V(T)$,
 and $\gamma$ maps each $t \in V(T)$ to a vertex-colored graph $\gamma(t) = (G_t,\chi_t)$ such that $\gamma(r) = (G,\chi)$, and for each leaf $t$ of $T$, it holds that $|V(G_t)| = 1$, and additionally, each internal node $t$ of $T$ satisfies (at least) one of the following:
 \begin{enumerate}[label=(D.\arabic*)]
  \item\label{item:refine} $t$ has exactly one child $s$ and $G_s = G_t$ (uncolored) and $\chi_s(v) = \WL{k}{G,\chi_t}(v)$ for all $v \in V(G_t)$,
  \item\label{item:flip} $t$ has exactly one child $s$ and $(G_s,\chi_s) = \flipf_f(G_t,\chi_t)$ for some flip function $f$ of $(G_t,\chi_t)$,
  \item\label{item:component} $G_t$ has exactly $\ell$ connected components with vertex sets $A_1,\dots,A_\ell$, and $t$ has exactly $\ell$ children $s_1,\dots,s_\ell$ and it holds that $(G_{s_i},\chi_{s_i}) = (G_t[A_i],\chi_t|_{A_i})$ for all $i \in [\ell]$, or
  \item\label{item:individualize} $t$ has exactly one child $s$ and $G_s = G_t$ and $\chi_s = \chi_t[u]$ for some $u \in V(G_t)$.
 \end{enumerate}
\end{definition}

There may be situations where an internal node $t$ satisfies more than one of the above conditions.
We say the node $t$ is an \emph{individualization node} if it satisfies Condition \ref{item:individualize} and none of the other conditions.
We define the \emph{individualization depth} of $(T,r,\gamma)$ to be the maximum number of individualization nodes on any root-to-leaf path in $(T,r)$.
The \emph{$k$-WL depth of $(G,\chi)$} is the minimum individualization depth of a $k$-IRC tree of $(G,\chi)$ and we denote it by $\WLd{k}{G,\chi}$.

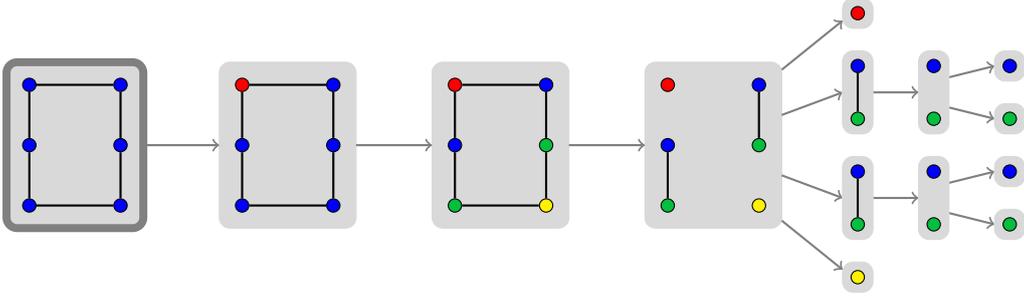
\begin{figure*}
 \centering
 \begin{tikzpicture}
  \draw[gray,line width=3pt,fill=gray!30,rounded corners] (-0.3,-0.3) rectangle (1.5,1.9);
  \draw[gray!30,fill=gray!30,rounded corners] (2.5,-0.3) rectangle (4.3,1.9);
  \draw[gray!30,fill=gray!30,rounded corners] (5.3,-0.3) rectangle (7.1,1.9);
  \draw[gray!30,fill=gray!30,rounded corners] (8.1,-0.3) rectangle (9.9,1.9);

  \draw[thick,gray,->] (1.5,0.8) -- (2.5,0.8);
  \draw[thick,gray,->] (4.3,0.8) -- (5.3,0.8);
  \draw[thick,gray,->] (7.1,0.8) -- (8.1,0.8);

  \draw[thick,gray,->] (9.9,1.8) -- (10.7,2.45);
  \draw[thick,gray,->] (9.9,1.2) -- (10.7,1.5);
  \draw[thick,gray,->] (9.9,0.4) -- (10.7,0.1);
  \draw[thick,gray,->] (9.9,-0.2) -- (10.7,-0.85);

  \draw[thick,gray,->] (11.1,1.5) -- (11.7,1.5);
  \draw[thick,gray,->] (11.1,0.1) -- (11.7,0.1);

  \draw[thick,gray,->] (12.1,1.3) -- (12.7,1.15);
  \draw[thick,gray,->] (12.1,1.7) -- (12.7,1.85);
  \draw[thick,gray,->] (12.1,0.3) -- (12.7,0.45);
  \draw[thick,gray,->] (12.1,-0.1) -- (12.7,-0.25);

  \node[vertex,fill=blue] (1-1) at (0,0) {};
  \node[vertex,fill=blue] (1-2) at (0,0.8) {};
  \node[vertex,fill=blue] (1-3) at (0,1.6) {};
  \node[vertex,fill=blue] (1-4) at (1.2,0) {};
  \node[vertex,fill=blue] (1-5) at (1.2,0.8) {};
  \node[vertex,fill=blue] (1-6) at (1.2,1.6) {};

  \draw[edge] (1-1) -- (1-2);
  \draw[edge] (1-1) -- (1-4);
  \draw[edge] (1-2) -- (1-3);
  \draw[edge] (1-3) -- (1-6);
  \draw[edge] (1-4) -- (1-5);
  \draw[edge] (1-5) -- (1-6);

  \node[vertex,fill=blue] (2-1) at (2.8,0) {};
  \node[vertex,fill=blue] (2-2) at (2.8,0.8) {};
  \node[vertex,fill=red] (2-3) at (2.8,1.6) {};
  \node[vertex,fill=blue] (2-4) at (4.0,0) {};
  \node[vertex,fill=blue] (2-5) at (4.0,0.8) {};
  \node[vertex,fill=blue] (2-6) at (4.0,1.6) {};

  \draw[edge] (2-1) -- (2-2);
  \draw[edge] (2-1) -- (2-4);
  \draw[edge] (2-2) -- (2-3);
  \draw[edge] (2-3) -- (2-6);
  \draw[edge] (2-4) -- (2-5);
  \draw[edge] (2-5) -- (2-6);

  \node[vertex,fill=darkpastelgreen] (3-1) at (5.6,0) {};
  \node[vertex,fill=blue] (3-2) at (5.6,0.8) {};
  \node[vertex,fill=red] (3-3) at (5.6,1.6) {};
  \node[vertex,fill=yellow] (3-4) at (6.8,0) {};
  \node[vertex,fill=darkpastelgreen] (3-5) at (6.8,0.8) {};
  \node[vertex,fill=blue] (3-6) at (6.8,1.6) {};

  \draw[edge] (3-1) -- (3-2);
  \draw[edge] (3-1) -- (3-4);
  \draw[edge] (3-2) -- (3-3);
  \draw[edge] (3-3) -- (3-6);
  \draw[edge] (3-4) -- (3-5);
  \draw[edge] (3-5) -- (3-6);

  \node[vertex,fill=darkpastelgreen] (4-1) at (8.4,0) {};
  \node[vertex,fill=blue] (4-2) at (8.4,0.8) {};
  \node[vertex,fill=red] (4-3) at (8.4,1.6) {};
  \node[vertex,fill=yellow] (4-4) at (9.6,0) {};
  \node[vertex,fill=darkpastelgreen] (4-5) at (9.6,0.8) {};
  \node[vertex,fill=blue] (4-6) at (9.6,1.6) {};

  \draw[edge] (4-1) -- (4-2);
  \draw[edge] (4-5) -- (4-6);

  \draw[gray!30,fill=gray!30,rounded corners] (10.9 - 0.2,2.55 - 0.2) rectangle (10.9 + 0.2,2.55 + 0.2);
  \draw[gray!30,fill=gray!30,rounded corners] (10.9 - 0.2,1.15 - 0.2) rectangle (10.9 + 0.2,1.85 + 0.2);
  \draw[gray!30,fill=gray!30,rounded corners] (10.9 - 0.2,-0.25 - 0.2) rectangle (10.9 + 0.2,0.45 + 0.2);
  \draw[gray!30,fill=gray!30,rounded corners] (10.9 - 0.2,-0.95 - 0.2) rectangle (10.9 + 0.2,-0.95 + 0.2);

  \node[vertex,fill=darkpastelgreen] (5-1) at (10.9,1.15) {};
  \node[vertex,fill=red] (5-3) at (10.9,2.55) {};
  \node[vertex,fill=blue] (5-5) at (10.9,1.85) {};

  \node[vertex,fill=darkpastelgreen] (5-2) at (10.9,-0.25) {};
  \node[vertex,fill=yellow] (5-4) at (10.9,-0.95) {};
  \node[vertex,fill=blue] (5-6) at (10.9,0.45) {};

  \draw[edge] (5-1) -- (5-5);
  \draw[edge] (5-2) -- (5-6);

  \draw[gray!30,fill=gray!30,rounded corners] (11.9 - 0.2,1.15 - 0.2) rectangle (11.9 + 0.2,1.85 + 0.2);
  \draw[gray!30,fill=gray!30,rounded corners] (11.9 - 0.2,-0.25 - 0.2) rectangle (11.9 + 0.2,0.45 + 0.2);

  \node[vertex,fill=darkpastelgreen] (6-1) at (11.9,1.15) {};
  \node[vertex,fill=blue] (6-5) at (11.9,1.85) {};

  \node[vertex,fill=darkpastelgreen] (6-2) at (11.9,-0.25) {};
  \node[vertex,fill=blue] (6-6) at (11.9,0.45) {};

  \foreach \y in {1.15,1.85,-0.25,0.45}{
   \draw[gray!30,fill=gray!30,rounded corners] (12.9 - 0.2,\y - 0.2) rectangle (12.9 + 0.2,\y + 0.2);
  }

  \node[vertex,fill=darkpastelgreen] (7-1) at (12.9,1.15) {};
  \node[vertex,fill=blue] (7-5) at (12.9,1.85) {};

  \node[vertex,fill=darkpastelgreen] (7-2) at (12.9,-0.25) {};
  \node[vertex,fill=blue] (7-6) at (12.9,0.45) {};

 \end{tikzpicture}
 \caption{A $1$-IRC tree of $C_6$.
  Every gray region corresponds to one node $t$ of the tree with the colored graph $\gamma(t)$ drawn inside.
  The root node is the only individualization node.
  Hence, the individualization depth of the $1$-IRC tree is equal to $1$.}
 \label{fig:example}
\end{figure*}

We start by determining the WL depth of some simple examples.

\begin{example}
 A $1$-IRC tree of the $6$-cycle $C_6$ (where every vertex receives the same color) is shown in Figure \ref{fig:example}.
 The $1$-IRC tree has individualization depth $1$, and hence $\WLd{1}{C_6} \leq 1$.
 Actually, it is not difficult to see that $\WLd{1}{C_6} = 1$.
\end{example}

\begin{example}
 \label{exa:complete-graph}
 Every vertex-colored complete graph $(K_n,\chi)$ has $k$-WL depth $0$ for all $k \geq 1$.
 Indeed, the graph $\flip(K_n,\chi)$ has no edges.
 So after splitting $\flip(K_n,\chi)$ into its connected components, we obtain that each one of them has only one vertex and thus satisfies the condition to be a leaf.
\end{example}

\begin{example}
 \label{exa:discrete-coloring}
 Let $(G,\chi)$ be a graph such that $\chi$ is a discrete coloring.
 Then $\flip(G,\chi)$ has no edges, so the graph can be split into connected components of size $1$.
 This implies that $\WLd{k}{G,\chi} = 0$ for all $k \geq 1$.

 More generally, if applying $1$-WL to a colored graph $(G,\chi)$ results in a discrete coloring, then $\WLd{k}{G,\chi} = 0$ for all $k \geq 1$.
 In particular, $\WLd{1}{G,\chi} = 0$ for almost all graphs $(G,\chi)$, see \cite{BabaiES80}.
\end{example}

The following lemma will turn out to be helpful when proving upper bounds on the WL depth.

\begin{lemma}
 \label{la:wl-depth-bound}
 Let $(G,\chi)$ be a colored graph and let $k \in \mathbb{N}$.
 Then
 \begin{enumerate}[label=(\arabic*)]
  \item\label{item:wl-depth-bound-1} $\WLd{k}{G,\chi} \leq \WLd{k}{G,\chi'}$ for the coloring $\chi'$ with $\chi'(v) = \WL{k}{G,\chi}(v)$ for all $v \in V(G)$;
  \item\label{item:wl-depth-bound-2} $\WLd{k}{G,\chi} \leq \WLd{k}{\flipf_f(G,\chi)}$ for every flip function $f$ of $(G,\chi)$;
  \item\label{item:wl-depth-bound-3} $\WLd{k}{G,\chi} \leq \max_{i \in [\ell]}\WLd{k}{G[A_i],\chi|_{A_i}}$, where $A_1,\dots,A_\ell$ are the vertex sets of the connected components of $G$; and
  \item\label{item:wl-depth-bound-4} $\WLd{k}{G,\chi} \leq 1 + \WLd{k}{G,\chi[u]}$ for every $u \in V(G)$.
 \end{enumerate}
\end{lemma}

\begin{proof}
 For Item \ref{item:wl-depth-bound-1}, let $(T',r',\gamma')$ be a $k$-IRC tree of $(G,\chi')$ of individualization depth $d$.
 We construct a $k$-IRC tree $(T,r,\gamma)$ by setting $V(T) \coloneqq V(T') \cup \{r\}$ (we assume $r \notin V(T')$), $E(T) \coloneqq E(T') \cup \{(r,r')\}$ and $\gamma(r) \coloneqq (G,\chi)$ and $\gamma(t) \coloneqq \gamma'(t)$ for all $t \in V(T')$.
 Note that $r$ satisfies Option \ref{item:refine} in Definition \ref{def:irc}, so $(T,r,\gamma)$ is a $k$-IRC tree of $(G,\chi)$ with individualization depth $d$.

 For Item \ref{item:wl-depth-bound-2}, let $f$ be a flip function of $(G,\chi)$.
 Let $(T',r',\gamma')$ be a $k$-IRC tree of $(G',\chi') \coloneqq \flipf_f(G,\chi)$ of individualization depth $d$.
 We construct a $k$-IRC tree $(T,r,\gamma)$ by setting $V(T) \coloneqq V(T') \cup \{r\}$ (we assume $r \notin V(T')$), $E(T) \coloneqq E(T') \cup \{(r,r')\}$ and $\gamma(r) \coloneqq (G,\chi)$ and $\gamma(t) \coloneqq \gamma'(t)$ for all $t \in V(T')$.
 Note that $r$ satisfies Option \ref{item:flip} in Definition \ref{def:irc}, so $(T,r,\gamma)$ is a $k$-IRC tree of $(G,\chi)$ of individualization depth $d$.

 For Item \ref{item:wl-depth-bound-3}, let $A_1,\dots,A_\ell$ denote the vertex sets of the connected components of $G$.
 For every $i \in [\ell]$, let $(T_i,r_i,\gamma_i)$ be a $k$-IRC tree of $(G_i,\chi_i)$ of individualization depth $d_i$.
 Without loss of generality, assume that $V(T_i) \cap V(T_j) = \emptyset$ and $r \notin V(T_i)$ for all distinct $i,j \in [\ell]$.
 We build a $k$-IRC tree $(T,r,\gamma)$ of $(G,\chi)$ with
 \[V(T) \coloneqq \{r\} \cup \bigcup_{i \in [\ell]} V(T_i)\]
 and
 \[E(T) \coloneqq \{(r,r_i) \mid i \in [\ell]\} \cup \bigcup_{i \in [\ell]} E(T_i).\]
 Also, we set $\gamma(r) \coloneqq (G,\chi)$ and $\gamma(t) \coloneqq \gamma_i(t)$ for all $i \in [\ell]$ and $t \in V(T_i)$.
 Note that $r$ satisfies Option \ref{item:component} in Definition \ref{def:irc}, so $(T,r,\gamma)$ is a $k$-IRC tree of $(G,\chi)$ of individualization depth $\max_{i \in [\ell]} d_i$.

 Finally, for Item \ref{item:wl-depth-bound-4}, let $u \in V(G)$.
 Also let $(T',r',\gamma')$ be a $k$-IRC tree of $(G,\chi')$ of individualization depth $d$ where $\chi' \coloneqq \chi[u]$.
 We construct a $k$-IRC tree $(T,r,\gamma)$ by setting $V(T) \coloneqq V(T') \cup \{r\}$ (we assume $r \notin V(T')$), $E(T) \coloneqq E(T') \cup \{(r,r')\}$ and $\gamma(r) \coloneqq (G,\chi)$ and $\gamma(t) \coloneqq \gamma'(t)$ for all $t \in V(T')$.
 Note that $r$ satisfies Option \ref{item:individualize} in Definition \ref{def:irc}, so $(T,r,\gamma)$ is a $k$-IRC tree of $(G,\chi)$ with individualization depth at most $1+d$.
\end{proof}

To demonstrate the potential of the concept of the WL depth, we show that forests have $k$-WL depth $0$ for every $k \in \mathbb{N}$.
 
\begin{lemma}
 Let $F$ be a forest and $\chi$ be a vertex coloring of $F$.
 Then $\WLd{k}{F,\chi} = 0$ for all $k \geq 1$.
\end{lemma}

\begin{proof}
 We show the statement by induction on the number of vertices of $F$.
 If $|V(F)| = 1$, then $\WLd{k}{F,\chi} = 0$ immediately follows.

 So suppose $|V(F)| > 1$.
 If $F$ is not connected, we split $F$ into its connected components $F_1,\dots,F_\ell$.
 By the induction hypothesis, $\WLd{k}{F_i,\chi|_{V(F_i)}} = 0$ for every $i \in [\ell]$.
 Then Lemma \ref{la:wl-depth-bound}\ref{item:wl-depth-bound-3} implies that
 \[\WLd{k}{F,\chi} \leq \max_{i \in [\ell]} \WLd{k}{F_i,\chi|_{V(F_i)}} = 0.\]

 So suppose $F$ is connected.
 Let $\chi^*(v) \coloneqq \WL{k}{F,\chi}(v)$ for all $v \in V(F)$.

 \begin{claim}
  \label{claim:disconnected-forest}
  $\flip(F,\chi)$ is a disconnected forest.
 \end{claim}

 \begin{claimproof}
  Note that it suffices to show that $E(\flip(F,\chi^*)) \subsetneq E(F)$.

  Consider two colors $c_1,c_2 \in \im(\chi^*)$.
  First suppose that $c_1 \neq c_2$.
  Since $\chi^*$ is stable with respect to $1$-WL, we conclude that the set $F[(\chi^*)^{-1}(c_1),(\chi^*)^{-1}(c_2)]$ is a disjoint union of $k$ stars $K_{1,h}$ for some $k \geq 1$ and $h \geq 0$ (for $h = 0$ the graph $K_{1,h}$ is an isolated vertex).
  If $k = 1$, then flipping removes all edges in $E_F(\chi^{-1}(c_1),\chi^{-1}(c_2))$.
  Otherwise, it has no effect on the edges in $E_F(\chi^{-1}(c_1),\chi^{-1}(c_2))$.

  So consider the case $c_1 = c_2$.
  Then, by regularity, $F[(\chi^*)^{-1}(c_1)]$ either has no edges or is a matching.
  If $F[(\chi^*)^{-1}(c_1)]$ is isomorphic to $K_2$, flipping removes the corresponding edge.
  Otherwise, flipping has no effect on the edges contained in $F[(\chi^*)^{-1}(c_1)]$.

  Hence, $E(\flip(F,\chi^*)) \subseteq E(F)$. We now show that at least one edge of $F$ is not present in $\flip(F,\chi^*)$. 

  Consider the graph repeatedly obtained by removing all leaves from $F$ until at most two vertices remain.
  These two vertices form a union of color classes under $\chi^*$ since $\chi^*$ is stable with respect to $1$-WL.
  If only one vertex remains, it forms a singleton color class and all incident edges are removed by flipping the graph. Otherwise, the two vertices are connected by an edge (since $F$ is connected), which is removed by flipping the graph.
 \end{claimproof}

 We thus have
 \[\WLd{k}{F,\chi} \stackrel{\text{L.\ }\ref{la:wl-depth-bound}\ref{item:wl-depth-bound-1}}{\leq} \WLd{k}{F,\chi^*}
                   \stackrel{\text{L.\ }\ref{la:wl-depth-bound}\ref{item:wl-depth-bound-2}}{\leq} \WLd{k}{\flip(F,\chi^*)} \leq 0\]
 where the last inequality follows as in the case that $F$ is disconnected.
\end{proof}

Next we prove that we can always refine colorings without increasing the WL depth.

\begin{lemma}
 \label{la:refine-coloring-smaller-wl-depth}
 Let $(G,\chi)$ be a vertex-colored graph and $k \geq 1$.
 Also let $\chi'\colon V(G) \to C$ be a vertex coloring such that $\chi' \preceq \chi$.
 Then $\WLd{k}{G,\chi'} \leq \WLd{k}{G,\chi}$.

 Additionally, for every $k$-IRC tree $(T,r,\gamma)$ of $(G,\chi)$ of individualization depth $d$, there is a $k$-IRC tree $(T',r',\gamma')$ of $(G,\chi')$ of individualization depth at most $d$ such that $(T',r') = (T,r)$.
\end{lemma}

\begin{proof}
 We prove the second part of the lemma; this clearly implies the first statement.

 Let $(T,r,\gamma)$ be a $k$-IRC tree of $(G,\chi)$ of individualization depth $d$.
 Also, for every $t \in V(T)$ let $(G_t,\chi_t) \coloneqq \gamma(t)$.

 For every $t \in V(T)$, we construct $\gamma'(t) = (G_t,\chi_t')$ so that $(T,r,\gamma')$ is a $k$-IRC tree of $(G,\chi')$ of individualization depth at most $d$.
 Additionally, we ensure that $\chi_t' \preceq \chi_t$ for every $t \in V(T)$.

 We describe the mapping $\gamma'$ in a top-down manner starting at the root.
 We set $\gamma'(r) \coloneqq (G,\chi')$.

 Now let $v \in V(T)$ be an internal node such that $\gamma'(t) = (G_t,\chi_t')$ has already been defined.

 First suppose $t$ satisfies Option \ref{item:refine} in Definition \ref{def:irc} in the $k$-IRC tree $(T,r,\gamma)$.
 Let $s$ be the unique child of $t$.
 Then $\chi_s(v) = \WL{k}{G,\chi_t}(v)$ for all $v \in V(G_t) = V(G_s)$.
 We set $\gamma'(s) \coloneqq (G_s,\chi_s')$ where $\chi_s'(v) = \WL{k}{G,\chi_t'}(v)$ for all $v \in V(G_t)$.
 Then $t$ also satisfies Option \ref{item:refine} in Definition \ref{def:irc} in the $k$-IRC tree $(T,r,\gamma')$.
 Also, since $\chi_t' \preceq \chi_t$, we conclude that $\chi_s' \preceq \chi_s$ by the properties of $k$-WL.

 Next suppose $t$ satisfies Option \ref{item:flip} in Definition \ref{def:irc} in the $k$-IRC tree $(T,r,\gamma)$.
 Let $s$ be the unique child of $t$.
 Then $(G_s,\chi_s) = \flipf_f(G_t,\chi_t)$ for some flip function $f$ of $(G_t,\chi_t)$.
 Let $C' \coloneqq \im(\chi_t')$.
 We define a flip function $f'$ of $(G_t,\chi_t')$ as follows.
 Let $c_1',c_2' \in C'$ be two colors.
 Since $\chi_t' \preceq \chi_t$, there are unique colors $c_1,c_2 \in \im(\chi_t)$ such that $(\chi_t')^{-1}(c_1') \subseteq \chi_t^{-1}(c_1)$ and $(\chi_t')^{-1}(c_2') \subseteq \chi_t^{-1}(c_2)$.
 We set $f'(c_1',c_2') \coloneqq f(c_1,c_2)$.
 Let $(G_s',\chi_s') \coloneqq \flipf_{f'}(G_t,\chi_t')$.
 Then $G_s' = G_s$.
 We set $\gamma'(s) \coloneqq (G_s,\chi_s')$.
 Then $t$ satisfies Option \ref{item:flip} in Definition \ref{def:irc} in the $k$-IRC tree $(T,r,\gamma')$ using the flip function $f'$.
 Also, $\chi_s' = \chi_t' \preceq \chi_t = \chi_s$.

 So suppose $t$ satisfies Option \ref{item:component} in Definition \ref{def:irc} in the $k$-IRC tree $(T,r,\gamma)$.
 Let $s_1,\dots,s_\ell$ denote the children of $t$.
 We set $\gamma'(s_i) \coloneqq (G_{s_i},\chi_t'|_{V(G_{s_i})})$ for every $i \in [\ell]$.
 Then $t$ satisfies Option \ref{item:component} in Definition \ref{def:irc} in the $k$-IRC tree $(T,r,\gamma')$.
 Also, $\chi_{s_i}' = \chi_t'|_{V(G_{s_i})} \preceq \chi_t|_{V(G_{s_i})} = \chi_{s_i}$ for every $i \in [\ell]$.

 Finally, suppose $t$ satisfies Option \ref{item:individualize} in Definition \ref{def:irc} in the $k$-IRC tree $(T,r,\gamma)$.
 Let $s$ be the unique child of $t$.
 Then $(G_s,\chi_s) = (G_t,\chi_t[u])$ for some $u \in V(G)$.
 We set $\chi_s' \coloneqq \chi_t'[u]$ and $\gamma'(s) \coloneqq (G_s,\chi_s')$.
 Then $t$ satisfies Option \ref{item:individualize} in Definition \ref{def:irc} in the $k$-IRC tree $(T,r,\gamma')$.
 Also, $\chi_s' = \chi_t'[u] \preceq \chi_t[u] = \chi_s$ since $\chi_t' \preceq \chi_t$.

 Since every internal node of $(T,r,\gamma)$ satisfies one of Options \ref{item:refine}-\ref{item:individualize} in Definition \ref{def:irc}, this completes the construction of $(T,r,\gamma')$.
 We obtain that $(T,r,\gamma')$ is a $k$-IRC tree of $(G,\chi')$.
 Also, every individualization node in $(T,r,\gamma')$ is also an individualization node in $(T,r,\gamma)$.
 So $(T,r,\gamma')$ has individualization depth at most $d$.
\end{proof}

Next, we argue that we can eliminate vertices with identical neighborhoods in the graph without increasing the $k$-WL depth. This will become an important tool to prove the bounds on the WL dimension that we want to obtain.

Let $(G,\chi)$ be a colored graph.
We say two vertices $v,w \in V(G)$ are \emph{strong twins} if $N_G[v] = N_G[w]$ (in particular, if $v \neq w$, then $vw \in E(G)$).
Also, we say that $v,w \in V(G)$ are \emph{weak twins} if $N_G(v) = N_G(w)$ (in particular, if $v \neq w$, then $vw \notin E(G)$).
Clearly, both relations are equivalence relations in $G$.
A partition $\pi$ of $V(G)$ is a \emph{twin partition of $G$} if $\pi$ is the partition into equivalence classes of either the strong-twin relation or the weak-twin relation in $G$.
Let $\pi$ be a twin partition of $G$.
We define the colored graph $(G/\pi,\chi/\pi)$ where $V(G/\pi) = \pi$,
\[E(G/\pi) = \{PP' \mid P,P' \in \pi, E_G(P,P') \neq \emptyset\}\]
and $(\chi/\pi)(P) = \{\!\{\chi(v) \mid v \in P\}\!\}$. 
Finally, we say that $G$ is \emph{twin-free} if there is no non-trivial twin partition of $G$, i.e., there are no distinct vertices $v,w \in V(G)$ that are strong or weak twins.

Note that, by the definition of twins, every edge $PP'$ in $G/\pi$ corresponds to a complete bipartite graph between the elements of $P$ and the elements of $P'$ in $G$.
Also note that for all $P,P' \in V(G/\pi)$ with $(\chi/\pi)(P) \neq (\chi/\pi)(P')$ and every $v \in P, v' \in P'$, it holds that $\WL{k}{G,\chi}(v) \neq \WL{k}{G,\chi}(v')$, since for $k \geq 2$, the algorithm $k$-WL distinguishes vertices for which the multisets of colors of siblings in the twin relation are not equal.

\begin{lemma}
 \label{la:remove-twins}
 Let $k \geq 2$.
 Let $(G,\chi)$ be a colored graph and let $\pi$ be a twin partition of $G$.
 Then
 \[\WLd{k}{G,\chi} \leq \WLd{k}{G/\pi,\chi/\pi}.\]
\end{lemma}

\begin{proof}
 Let $k \geq 2$ be fixed.
 We show that for every colored graph $(G,\chi)$, every twin partition $\pi$ of $G$, and every $k$-IRC tree $(T',r',\gamma')$ of $(G',\chi') \coloneqq (G/\pi,\chi/\pi)$ of individualization depth $d$, it holds that $\WLd{k}{G,\chi} \leq d$.
 We prove this statement by induction on the height $h$ of the tree $(T',r')$.

 In the base case, $h = 0$.
 Then $|V(G')| = 1$, which means that all vertices in $(G,\chi)$ are twins. Then $G$ is a complete graph or has no edges and in both cases, $\WLd{k}{G,\chi}=0 \leq d$ (see Example \ref{exa:complete-graph}).

 So as the inductive hypothesis, assume the statement holds for all graphs $(H,\lambda)$ and twin partitions $\sigma$ of $H$ for which $(H/\sigma,\lambda/\sigma)$ has a $k$-IRC tree of height less than $h$. Let $(G,\chi)$ be a colored graph for which $(G',\chi') \coloneqq (G/\pi,\chi/\pi)$ has a $k$-IRC tree $(T',r',\gamma')$ of height $h$. Let $d$ be the individualization depth of $(T',r',\gamma')$.

 Let $t'_1,\dots,t'_\ell$ denote the children of $r'$. For $i \in [\ell]$, let $(G'_i,\chi'_i) \coloneqq \gamma'(t'_i)$ and let $T'_i$ denote the subtree of $T'$ rooted at $t'_i$.
 Observe that $(T'_i,t'_i,\gamma'|_{V(T'_i)})$ is a $k$-IRC tree of $(G'_i,\chi'_i)$ and has height less than $h$. Let $d_i$ denote the individualization depth of $(T'_i,t'_i,\gamma'|_{V(T'_i)})$. Observe that $d_i \leq d$, and if $r'$ is an individualization node, then it even holds that $d_i \leq d-1$.

 First suppose that $\ell = 1$, $G'_1 = G'$, and $\chi'_1(v) = \WL{k}{G',\chi'}(v)$ for all $v \in V(G')$.
 Let $\chi_1(v) \coloneqq \WL{k}{G,\chi}(v)$ for all $v \in V(G)$.
 Note that $\chi_1 / \pi \equiv \chi'_1$ since $k \geq 2$.
 Furthermore, we know that $(G'_1,\chi'_1)$ has a $k$-IRC tree of height less than $h$ and individualization depth at most $d$ (for instance $(T'_1,t'_1,\gamma'|_{V(T'_1)})$).
 Thus, by Lemma \ref{la:refine-coloring-smaller-wl-depth}, the same holds for $(G',\chi_1 / \pi)$.
 Hence, we can apply the induction hypothesis to $(G,\chi_1)$ and $\pi$ and obtain that
 \begin{alignat*}{2}
  \WLd{k}{G,\chi} &\alignsymbols{\stackrel{\text{L.\ }\ref{la:wl-depth-bound}\ref{item:wl-depth-bound-1}}{\leq}} \WLd{k}{G,\chi_1} \mspace{-9mu} && \alignsymbols{\stackrel{\text{IH}}{\leq}} \mspace{-9mu}\WLd{k}{G',\chi_1/\pi}\\
                  &\alignsymbols{\stackrel{\text{L.\ }\ref{la:refine-coloring-smaller-wl-depth}}{=}} \WLd{k}{G',\chi'_1} \mspace{-9mu}&& \alignsymbols{\leq}\mspace{-9mu} d.
 \end{alignat*}

 Next, suppose that $\ell = 1$ and $(G'_1,\chi'_1) = \flipf_{f'}(G',\chi')$ for some flip function $f'$ of $(G',\chi')$. Let $\chi_1(v) \coloneqq \WL{k}{G,\chi}(v)$ for all $v \in V(G)$. Let $C \coloneqq \im(\chi_1)$.
 We define a flip function $f$ of $(G,\chi_1)$ as follows.
 Let $c_1,c_2 \in C$ be two colors. Since the sets of colors that $k$-WL computes for vertices (in $V(G)$) contained in differently colored $P,P' \in V(G')$, there are unique colors $c'_1,c'_2 \in \im(\chi')$ such that
 \begin{align*}
 \chi_1^{-1}(c_1) \subseteq &\bigcup_{P \in (\chi')^{-1}(c_1')} P,
 \\\chi_1^{-1}(c_2) \subseteq &\bigcup_{P \in (\chi')^{-1}(c_2')} P.
 \end{align*}
 We set $f(c_1,c_2) \coloneqq f'(c'_1,c'_2)$.
 Let $G_1$ be the first component of $\flipf_{f}(G,\chi_1)$.
 Then $G_1/\pi = G'_1$, and also $\chi_1/\pi \preceq \chi'_1 = \chi'$.
 Furthermore, we know that $(G'_1,\chi'_1)$ has a $k$-IRC tree of height less than $h$ and individualization depth at most $d$ (for instance $(T'_1,t'_1,\gamma'|_{V(T'_1)})$).
 Thus, by Lemma \ref{la:refine-coloring-smaller-wl-depth}, the same holds for $(G'_1,\chi_1 / \pi)$.
 Hence, we can apply the induction hypothesis to $(G_1,\chi_1)$ and $\pi$ and obtain that
 \begin{alignat*}{2}
  \WLd{k}{G,\chi} &\alignsymbols{\stackrel{\text{L.\ }\ref{la:wl-depth-bound}\ref{item:wl-depth-bound-1}}{\leq}} \mspace{-4mu} \WLd{k}{G,\chi_1} &&\alignsymbols{\stackrel{\text{L.\ }\ref{la:wl-depth-bound}\ref{item:wl-depth-bound-2}}{\leq}} \mspace{-4mu}\WLd{k}{G_1,\chi_1}\\
                  &\alignsymbols{\stackrel{\text{IH}}{\leq}} \mspace{-4mu} \WLd{k}{G'_1,\chi_1/\pi} \mspace{-8mu} &&\alignsymbols{\stackrel{\text{L.\ } \ref{la:refine-coloring-smaller-wl-depth}}{\leq}} \mspace{-4mu}\WLd{k}{G'_1,\chi'_1} \ \leq \ d.
 \end{alignat*}

 Now, suppose that $G'$ is disconnected and consists of connected components with vertex sets $A'_1,\dots,A'_\ell$ and we have $(G'_i,\chi'_i) = (G'[A'_i],\chi'|_{A'_i})$ for all $i \in [\ell]$.
 Let $A_i$ be the vertex set of the connected component of $G$ corresponding to $A'_i$.
 Define $(G_i,\chi_i) \coloneqq (G[A_i],\chi|_{A_i})$.
 Note that $G_i / \pi = G'_i$ and also $\chi_i / \pi \equiv \chi'_i$ for every $i \in [\ell]$.
 Furthermore, we know that every $(G'_i,\chi'_i)$ has a $k$-IRC tree of height less than $h$ and individualization depth at most $d_i$ (for instance $(T'_i,t'_i,\gamma'|_{V(T'_i)})$).
 Thus, by Lemma \ref{la:refine-coloring-smaller-wl-depth}, the same holds for $(G'_i,\chi_i / \pi)$.
 Hence, we can apply the induction hypothesis to every $(G_i,\chi_i)$ and $\pi$ and obtain that

 \begin{alignat*}{2}
  \WLd{k}{G,\chi} &\alignsymbols{\stackrel{\text{L.\ }\ref{la:wl-depth-bound}\ref{item:wl-depth-bound-3}}{\leq}} \max_{i \in [\ell]}\WLd{k}{G_i,\chi_i}
                  &&\alignsymbols{\stackrel{\text{IH}}{\leq}} \max_{i \in [\ell]} \WLd{k}{G'_i,\chi_i/\pi}\\
                  &\alignsymbols{\stackrel{\text{L.\ } \ref{la:refine-coloring-smaller-wl-depth}}{=}} \max_{i \in [\ell]} \WLd{k}{G'_i,\chi'_i} &&\alignsymbols{\leq} \max_{i \in [\ell]} d_i \ \leq \ d.
 \end{alignat*}

 Finally, suppose that $\ell = 1$ and there is some $P \in V(G')$ such that $(G'_1,\chi'_1) = (G', \chi'[P])$.
 Choose any $u \in P$ and define $(G_1,\chi_1) \coloneqq (G,\chi[u])$.
 Note that $G_1 / \pi = G'$ and also $\chi_1 / \pi \equiv \chi'_1$.
 Furthermore, we know that $(G',\chi'_1)$ has a $k$-IRC tree of height less than $h$ and individualization depth at most $d-1$ (for instance $(T'_i,t'_i,\gamma'|_{V(T'_i)})$).
 Thus, by Lemma \ref{la:refine-coloring-smaller-wl-depth}, the same holds for $(G',\chi_1 / \pi)$.
 Hence, we can apply the induction hypothesis to $(G,\chi_1)$ and $\pi$ and obtain that
 \begin{alignat*}{2}
  \WLd{k}{G,\chi} &\alignsymbols{\stackrel{\text{L.\ }\ref{la:wl-depth-bound}\ref{item:wl-depth-bound-4}}{\leq}} 1 + \WLd{k}{G,\chi_1}
                  &&\alignsymbols{\stackrel{\text{IH}}{\leq}} 1 + \WLd{k}{G',\chi_1/\pi}\\
                  &\alignsymbols{\stackrel{\text{L.\ } \ref{la:refine-coloring-smaller-wl-depth}}{=}} 1 + \WLd{k}{G',\chi'_1} &&\alignsymbols{\leq} 1 + d-1 \ = \ d.
 \end{alignat*}
This concludes the proof.
\end{proof}

Our main motivation to study the $k$-WL depth of graphs is to obtain improved bounds on the $k$-WL dimension of graphs.
Towards this end, the following lemma relates the $k$-WL depth of a graph to its $k$-WL dimension.

\begin{lemma}
 \label{la:wl-depth-to-dimension}
 Let $k,d \in \mathbb{N}$ and let $(G,\chi)$ be a colored graph with $\WLd{k}{G,\chi} \leq d$.
 Then the $(\max\{2,k\}+d)$-dimensional Weisfeiler--Leman algorithm identifies $(G,\chi)$.
\end{lemma}

\begin{proof}
 Let $(T,r,\gamma)$ be a $k$-IRC tree of $(G,\chi)$ of individualization depth $d$. Let $h$ be the height of $(T,r)$. We prove the statement by induction on $h$.

 In the base case, $h = 0$.
 Then $|V(G)| = 1$ and $(G,\chi)$ is identified by $k$-WL for every $k \in \mathbb{N}$.

 So as the inductive hypothesis, assume the statement holds for all graphs $(H,\lambda)$ that have a $k$-IRC tree of height smaller than $h$ and of individualization depth at most $d$.

 Let $t_1,\dots,t_\ell$ denote the children of $r$.
 For $i \in [\ell]$, let $(G_i,\chi_i) \coloneqq \gamma(t_i)$ and let $T_i$ denote the subtree of $T$ rooted at $t_i$.
 Observe that $(T_i,t_i,\gamma|_{V(T_i)})$ is a $k$-IRC tree of $(G_i,\chi_i)$ of height smaller than $h$.
 Let $d_i$ denote the individualization depth of $(T_i,t_i,\gamma|_{V(T_i)})$. Observe that $d_i \leq d$, and if $r$ is an individualization node, then it even holds that $d_i \leq d-1$.

 First suppose that $\ell = 1$, $G_1 = G$, and $\chi_1(v) = \WL{k}{G,\chi}(v)$ for all $v \in V(G)$.
 By the induction hypothesis, $(G_1,\chi_1)$ is identified by $(\max\{2,k\}+d)$-WL.
 Hence, $(G,\chi)$ is identified by $(\max\{2,k\}+d)$-WL.

 Next, suppose that $\ell = 1$ and $(G_1,\chi_1) = \flipf_{f}(G,\chi)$ for some flip function $f$ of $(G,\chi)$.
 By the induction hypothesis, $(G_1,\chi_1)$ is identified by $(\max\{2,k\}+d)$-WL.
 Consider a colored graph $(G',\chi')$ with $(G',\chi') \simeq_{\max\{2,k\}+d} (G,\chi)$.
 Then it has to have the same vertex color classes and sizes as $(G,\chi)$.
 We conclude that $\flipf_f(G'\chi') \simeq_{\max\{2,k\}+d} (G_1,\chi_1)$, which implies that there is an isomorphism $\varphi\colon\flipf_f(G'\chi') \cong (G_1,\chi_1)$ (since $(G_1,\chi_1)$ is identified by $(\max\{2,k\}+d)$-WL).
 But then $\varphi$ is also an isomorphism from $(G',\chi')$ to $(G,\chi)$.
 Hence, $(G,\chi)$ is identified by $(\max\{2,k\}+d)$-WL.

 Now, suppose that $G$ is disconnected and consists of connected components with vertex sets $A_1,\dots,A_\ell$ and $(G_i,\chi_i) = (G[A_i],\chi|_{A_i})$ for all $i \in [\ell]$.
 By the induction hypothesis, each $(G[A_i],\chi|_{A_i})$ is identified by $(\max\{2,k\}+d)$-WL.
 Since $2$-WL distinguishes pairs of vertices in the same connected component from pairs of vertices in different connected components, the algorithm $(\max\{2,k\}+d)$-WL distinguishes $(G,\chi)$ from any graph whose multiset of isomorphism types of connected components is different.
 Hence, $(G,\chi)$ is identified by $(\max\{2,k\}+d)$-WL.
 
 Finally, suppose that $\ell = 1$ and that there is some $u \in V(G)$ such that $(G_1,\chi_1) = (G, \chi[u])$.
 By the induction hypothesis, $(G,\chi[u])$ is identified by $(\max\{2,k\}+d-1)$-WL. Hence, by Lemma \ref{la:indiv-dimension}, $(G,\chi)$ is identified by $(\max\{2,k\}+d)$-WL.
\end{proof}

In analyzing the $k$-WL depth of (vertex-colored) graphs $(G,\chi)$, the challenge is when only the last option is applicable, i.e., $(G,\chi)$ is connected and flipped, and $\chi$ is stable with respect to $k$-WL, that is, it is not refined by it.
In such a case, we say that $(G,\chi)$ is \emph{$k$-robust}.

\begin{lemma}
 \label{la:robust-no-singleton}
 Let $k \geq 1$ and suppose that $(G,\chi)$ is $k$-robust and $|V(G)| \geq 2$.
 Then $|\chi^{-1}(c)| \geq 2$ for every color $c \in \im(\chi)$.
\end{lemma}

\begin{proof}
 Suppose towards a contradiction that $|\chi^{-1}(c)| = 1$ holds for some color $c$, and let $v \in V(G)$ denote the unique vertex of color $c$.
 Since $G$ is connected and $|V(G)| \geq 2$, there exists a $w \in N(v)$.
 Let $D \coloneqq \{w' \in V(G) \mid \chi(w') = \chi(w)\}$.
 Since $\chi$ is stable with respect to $k$-WL, we conclude that $D \subseteq N(v)$.
 But this contradicts $(G,\chi)$ being flipped since all edges are present between the color classes $\{v\}$ and $D$.
\end{proof}

For a graph $G$, a vertex coloring $\chi$ and a set $U \subseteq V(G)$, we define the coloring $\chi_{k,U}$ via $\chi_{k,U}(v) \coloneqq \WL{k}{G,\chi[U]}(v)$ for all $v \in V(G)$.
For a single vertex $u \in V(G)$, we also set $\chi_{k,u} \coloneqq \chi_{k,\{u\}}$. 

Now, a simple strategy to bound the $k$-WL depth of a graph is, given a $k$-robust graph, to individualize a small set of vertices $U$ so that the number of colors in $\chi_{k,U}$ increases as much as possible.

\begin{lemma}
 \label{la:wl-depth-from-color-class-increase}
 Let $\xi,k \geq 1$ be integers.
 Suppose that, for every $k$-robust graph $(G,\chi)$ with $n > 1$ vertices, there is a set $U \subseteq V(G)$ such that $|\im(\chi_{k,U})| \geq |\im(\chi)| + \xi|U|$.
 Then
 \[\WLd{k}{G,\chi} \leq \frac{n-1}{\xi}\]
 holds for every colored graph $(G,\chi)$.
\end{lemma}

\begin{proof}
 We prove that
 \begin{equation}
  \label{eq:wl-depth-xi-simple}
  \WLd{k}{G,\chi} \leq \frac{|V(G)|-|\im(\chi)|}{\xi}
 \end{equation}
 holds for every colored graph $(G,\chi)$, which implies the statement.

 We prove \eqref{eq:wl-depth-xi-simple} by induction on the tuple $(|V(G)|,|\im(\chi)|,|E(G)|)$.
 Consider the set $M \coloneqq \{(n,\ell,m) \in \ZZ_{\geq 0}^3 \mid \ell \leq n\}$ and observe that $(|V(G)|,|\im(\chi)|,|E(G)|) \in M$.
 For the induction, we define a linear order $\prec$ on $M$ via $(n,\ell,m) \prec (n',\ell',m')$ if $n < n'$, or $n = n'$ and $\ell > \ell'$, or $n = n'$ and $\ell = \ell'$ and $m < m'$.
 Note that we use the inverse order on the second component, i.e., for $\ell \neq \ell'$, we have $(n,\ell,m) \prec (n,\ell',m')$ if $\ell > \ell'$.
 Still, since $\ell \leq n$ for every $(n,\ell,m) \in M$, there are no infinite decreasing chains in $M$.

 For the base case, suppose that $|V(G)| = 1$.
 Then $\WLd{1}{G,\chi} = 0 = \frac{|V(G)| - |\im(\chi)|}{\xi}$ and the statement holds.

 For the inductive step, suppose that $(G,\chi)$ is a colored graph with $|V(G)| > 1$.
 We distinguish several cases.
 \begin{itemize}[leftmargin=3ex]
  \item
   First, suppose that $\chi$ is not stable with respect to $k$-WL, i.e., $\chi^* \prec \chi$ where $\chi^*(v) \coloneqq \WL{k}{G,\chi}(v)$ for all $v \in V(G)$.
   Observe that $|\im(\chi^*)| > |\im(\chi)|$. This implies that $(|V(G)|,|\im(\chi^*)|,|E(G)|) \prec (|V(G)|,|\im(\chi)|,|E(G)|)$.
   Hence, by the induction hypothesis, we obtain
   \[\WLd{k}{G,\chi^*} \leq \frac{|V(G)|-|\im(\chi^*)|}{\xi}.\]
   Also, $\WLd{k}{G,\chi} \leq \WLd{k}{G,\chi^*}$ by Lemma \ref{la:wl-depth-bound}\ref{item:wl-depth-bound-1}.
   Together, we obtain that
   \begin{align*}
    \WLd{k}{G,\chi} \leq \WLd{k}{G,\chi^*} \leq \frac{|V(G)|-|\im(\chi^*)|}{\xi} \leq \frac{|V(G)|-|\im(\chi)|}{\xi}.
   \end{align*}
  \item
   Next, assume that $(G,\chi)$ is not flipped.
   Let $(G',\chi') \coloneqq \flip(G,\chi)$.
   Observe that $V(G') = V(G)$, $\chi' = \chi$ and $|E(G')| < |E(G)|$.
   So we conclude that  $(|V(G')|,|\im(\chi')|,|E(G')|) \prec (|V(G)|,|\im(\chi)|,|E(G)|)$.
   Hence, by the induction hypothesis, we get
   \[\WLd{k}{G',\chi'} \leq \frac{|V(G')|-|\im(\chi')|}{\xi}.\]
   Also, $\WLd{k}{G,\chi} \leq \WLd{k}{G',\chi'}$ by Lemma \ref{la:wl-depth-bound}\ref{item:wl-depth-bound-2}.
   Together, we obtain that
   \begin{align*}
    \WLd{k}{G,\chi} \leq \WLd{k}{G',\chi'} \leq \frac{|V(G')|-|\im(\chi')|}{\xi} = \frac{|V(G)|-|\im(\chi)|}{\xi}.
   \end{align*}
  \item
   Now suppose that $G$ is not connected and let $A_1,\dots,A_\ell$ be the vertex sets of the connected components of $G$.
   Observe that $|A_i| < |V(G)|$ for all $i \in [\ell]$.
   By the induction hypothesis, we have that
   \[\WLd{k}{G[A_i],\chi|_{A_i}} \leq \frac{|A_i|-|\im(\chi|_{A_i})|}{\xi}\]
   for all $i \in [\ell]$.
   Together with Lemma \ref{la:wl-depth-bound}\ref{item:wl-depth-bound-3}, we obtain
   \begin{align*}
    \WLd{k}{G,\chi} \leq \max_{i \in [\ell]} \WLd{k}{G[A_i],\chi|_{A_i}} \leq \max_{i \in [\ell]} \frac{|A_i| - |\im(\chi|_{A_i})|}{\xi}.
   \end{align*}
   Also, since $|A_i| - |\im(\chi|_{A_i})| \geq 0$, we get that
   \begin{align*}
    \max_{i \in [\ell]}(|A_i| - |\im(\chi|_{A_i})|) &\leq  \sum_{i \in [\ell]} (|A_i| - |\im(\chi|_{A_i})|)\\
                                                    &= |V(G)| - \sum_{i \in [\ell]}|\im(\chi|_{A_i})| \leq |V(G)| - |\im(\chi)|.
   \end{align*}
   So $\WLd{k}{G,\chi} \leq \frac{|V(G)|-|\im(\chi)|}{\xi}$.
  \item
   Finally, suppose that $(G,\chi)$ is connected and flipped, $|V(G)| > 1$, and $\chi$ is stable with respect to $k$-WL.
   By assumption, there is a $U \subseteq V(G)$ such that $|\im(\chi_{k,U})| \geq |\im(\chi)| + \xi|U|$.
   In particular, $|\im(\chi_{k,U})| > |\im(\chi)|$, hence $(|V(G)|,|\im(\chi_{k,U})|,|E(G)|) \prec (|V(G)|,|\im(\chi)|,|E(G)|)$.
   Thus, by the induction hypothesis, we have
   \[\WLd{k}{G,\chi_{k,U}} \leq \frac{|V(G)|-|\im(\chi_{k,U})|}{\xi}.\]
   Also, applying Lemma \ref{la:wl-depth-bound}\ref{item:wl-depth-bound-4} a total of $|U|$ many times and then applying Lemma \ref{la:wl-depth-bound}\ref{item:wl-depth-bound-1}, we obtain
   \begin{align*}
    \WLd{k}{G,\chi} \leq \WLd{k}{G,\chi[U]} + |U| \leq \WLd{k}{G,\chi_{k,U}} + |U|.
   \end{align*}
   Together, it follows that
   \begin{alignat*}{2}
        \WLd{k}{G,\chi} &\alignsymbols{\leq} \WLd{k}{G,\chi_{k,U}} + |U| &&\alignsymbols{\leq} \frac{|V(G)|-|\im(\chi_{k,U})|}{\xi} + |U|\\
                        &\alignsymbols{\leq} \frac{|V(G)| - (|\im(\chi)| + \xi|U|) + \xi|U|}{\xi} &&\alignsymbols{=} \frac{|V(G)| - |\im(\chi)|}{\xi}.
   \end{alignat*}
 \end{itemize}
 This concludes the proof.
\end{proof}

With this tool at hand, we can already obtain a first bound on the $1$-WL depth of arbritrary graphs.

For a vertex-colored graph $(G,\chi)$, we define the graph $G[[\chi]]$ with vertex set $V(G[[\chi]]) \coloneqq \im(\chi)$ and edge set
\[E(G[[\chi]]) \coloneqq \{(c_1,c_2) \mid E_G(\chi^{-1}(c_1),\chi^{-1}(c_2)) \neq \emptyset\}.\]
Here, we explicitly define $G[[\chi]]$ to contain self-loops, i.e., there is a loop $(c_1,c_1) \in E(G[[\chi]])$ if $E_G(\chi^{-1}(c_1),\chi^{-1}(c_1)) \neq \emptyset$.
For $c_1 \in V(G[[\chi]])$, we set $\deg_{G[[\chi]]}(c_1) \coloneqq |\{c_2 \in V(G[[\chi]]) \mid (c_1,c_2) \in E(G[[\chi]])\}|$.

\begin{lemma}
 \label{la:color-class-increase-general-robust}
 Let $(G,\chi)$ be a $1$-robust graph such that $|V(G)|> 1$.
 Let $c \in \im(\chi)$ and set $\xi \coloneqq \deg_{G[[\chi]]}(c) + 1$.
 Then $|\im(\chi_{1,u})| \geq |\im(\chi)| + \xi$ for every $u \in \chi^{-1}(c)$.
\end{lemma}

\begin{proof}
 Let $c_1,\dots,c_d$ be the neighbors of $c$ in the graph $G[[\chi]]$.
 We have that $|\im(\chi[u])| \geq |\im(\chi)| + 1$.
 Hence, it suffices to argue that $|\im(\chi_{1,u})| \geq |\im(\chi[u])| + d$.
 Let $i \in [d]$ and first suppose that $c_i \neq c$.
 Then the graph $G[\chi^{-1}(c),\chi^{-1}(c_i)]$ is biregular, because $\chi$ is stable with respect to $1$-WL.
 Moreover, this graph contains at least one edge since $(c,c_i) \in E(G[[\chi]])$, and it is not a complete bipartite graph since $(G,\chi)$ is flipped.
 So, after individualizing $u$ and performing $1$-WL, the color class $c_i$ is split.
 
 For $c = c_i$, the same argument applies.
 The graph $G[\chi^{-1}(c)]$ contains at least one edge, since $(c,c) \in E(G[[\chi]])$, and it is not a complete graph, since $(G,\chi)$ is flipped.
 So, after individualizing $u$ and performing $1$-WL, the color class $c$ is split one more time (following the split that occurs from individualizing $u$).
\end{proof}

\begin{corollary}
 \label{cor:wl-depth-robust}
 Let $(G,\chi)$ be a colored graph and let $n \coloneqq |V(G)|$. Then $\WLd{1}{G,\chi} \leq \frac{n-1}{2}$.
\end{corollary}

\begin{proof}
 We set $\xi \coloneqq 2$.
 Let $(G,\chi)$ be a $1$-robust graph such that $|V(G)|> 1$.
 Then $\deg_{G[[\chi]]}(c) \geq 1$ for all $c \in \im(\chi)$.
 So there is some vertex $u \in V(G)$ such that $|\im(\chi_{1,u})| \geq |\im(\chi)| + \xi$ by Lemma \ref{la:color-class-increase-general-robust}.

 So it holds that $\WLd{1}{G,\chi} \leq \frac{n-1}{2}$ for all colored graphs $(G,\chi)$ by Lemma \ref{la:wl-depth-from-color-class-increase}.
\end{proof}

In combination with Lemma \ref{la:wl-depth-to-dimension}, we obtain that the WL dimension of every $n$-vertex graph is at most $\frac{n+3}{2}$, which essentially recovers the result from \cite{PikhurkoVV06} using simpler arguments (but using counting quantifiers which are not needed in \cite{PikhurkoVV06}).
In the following, we extend this type of argument to obtain improved bounds.

\section{Bounding the WL Depth in the Vertex Cover Number}
\label{sec:vc}

In this section, we prove Theorem \ref{thm:main-vc}.
On a high level, the strategy for the proof is similar to the one for Corollary \ref{cor:wl-depth-robust}, in which we have argued that for every $1$-robust graph $(G,\chi)$, there is some vertex $u \in V(G)$ such that, after individualizing $u$ and performing $1$-WL, the number of colors in $G$ increases by at least $2$.

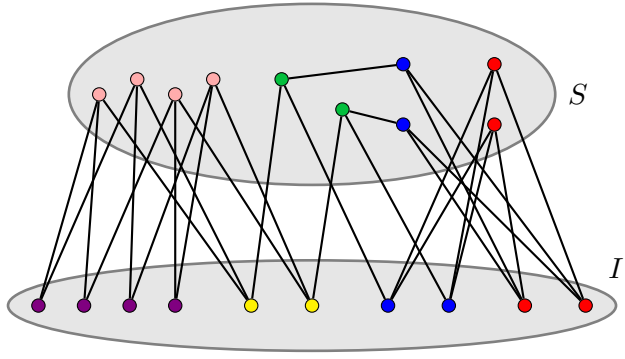
\begin{figure}
 \centering
 \begin{tikzpicture}

  \draw[gray,line width=1pt,fill=gray!20] (0,2.8) ellipse (3.2cm and 1.2cm);
  \node at (3.5,2.8) {$S$};
  \draw[gray,line width=1pt,fill=gray!20] (0,0) ellipse (4cm and 0.6cm);
  \node at (4.0,0.5) {$I$};

  \node[vertex,fill=red] (r1) at (2.4,3.2) {};
  \node[vertex,fill=red] (r2) at (2.4,2.4) {};
  \node[vertex,fill=red] (r3) at (2.8,0) {};
  \node[vertex,fill=red] (r4) at (3.6,0) {};

  \node[vertex,fill=blue] (b1) at (1.2,3.2) {};
  \node[vertex,fill=blue] (b2) at (1.2,2.4) {};
  \node[vertex,fill=blue] (b3) at (1.8,0) {};
  \node[vertex,fill=blue] (b4) at (1.0,0) {};

  \node[vertex,fill=darkpastelgreen] (g1) at (0.4,2.6) {};
  \node[vertex,fill=darkpastelgreen] (g2) at (-0.4,3.0) {};

  \node[vertex,fill=yellow] (y1) at (0,0) {};
  \node[vertex,fill=yellow] (y2) at (-0.8,0) {};

  \node[vertex,fill=pink!90!red] (o1) at (-1.3,3.0) {};
  \node[vertex,fill=pink!90!red] (o2) at (-1.8,2.8) {};
  \node[vertex,fill=pink!90!red] (o3) at (-2.3,3.0) {};
  \node[vertex,fill=pink!90!red] (o4) at (-2.8,2.8) {};

  \node[vertex,fill=violet] (v1) at (-1.8,0) {};
  \node[vertex,fill=violet] (v2) at (-2.4,0) {};
  \node[vertex,fill=violet] (v3) at (-3.0,0) {};
  \node[vertex,fill=violet] (v4) at (-3.6,0) {};

  \foreach \v/\w in {r1/r4,r2/r3,r1/b3,r1/b4,r2/b3,r2/b4,r3/b1,r3/b2,r4/b1,r4/b2,g1/b2,g1/b3,g2/b1,g2/b4,
                     g1/y1,g2/y2,y1/o1,y1/o2,y2/o3,y2/o4,v1/o1,v1/o2,v2/o1,v2/o3,v3/o2,v3/o4,v4/o3,v4/o4}{
   \draw[edge] (\v) -- (\w);
  }

 \end{tikzpicture}
 \caption{A colored graph $(G,\chi)$ and a vertex cover $S$ of $G$.
  We set $I \coloneqq V \setminus S$.
  The set $\chi(S)$ contains the colors \textsf{{\color{red}red}}, \textsf{{\color{blue}blue}}, \textsf{{\color{darkpastelgreen}green}} and \textsf{{\color{pink!90!red}pink}}.
  The set $\chi(S) \setminus \chi(I)$ contains only \textsf{{\color{darkpastelgreen}green}} and \textsf{{\color{pink!90!red}pink}}.
  Hence, $p = 4 + 2 = 6$ in this example.
  }
 \label{fig:vc-example}
\end{figure}

Let $(G,\chi)$ be a colored graph and let $S$ be a vertex cover of $G$.
An example is given in Figure~\ref{fig:vc-example}.
We intend to find a vertex $u \in V(G)$ such that, after individualizing $u$ and performing $1$-WL, the number of colors \emph{in the set $S$} grows as much as possible.
Indeed, as soon as all vertices in $S$ have pairwise different colors, one can prove easily that the $1$-WL depth of the graph is equal to $0$.

To actually implement this idea, we rely on some additional instruments.
First of all, note that a color appearing in the set $S$ may also appear outside of $S$ (i.e., there may be vertices $v \in S$ and $w \notin S$ with $\chi(v) = \chi(w)$).
Here, our intuition is that it is preferable if colors only appear in the set $S$.
This motivates us to consider a different progress measure which, instead of only counting colors in $S$, also gives additional weight to colors that only appear in $S$.
More precisely, we aim to increase the parameter $p \coloneqq |\chi(S)| + |\chi(S) \setminus \chi(V \setminus S)|$ as much as possible using a single individualization followed by performing $1$-WL.
Since the maximum value of $p$ is equal to $2|S|$, we need to increase $p$ by at least $3$ with a single individualization to achieve the desired bound.

To reach this goal, we distinguish two cases.
In the first case, $\chi(S) \cap \chi(V \setminus S) \neq \emptyset$ holds, and we prove that we have either reached the goal, or individualizing a single vertex and refining the coloring with $1$-WL renders the two color sets disjoint.
This is covered in Lemma \ref{la:vertex-cover-non-disjoint-colors} below.
Note that, as opposed to Lemma \ref{la:color-class-increase-general-robust}, it is not possible to restrict our analysis to $1$-robust graphs, since arbitrary flips in the edge set may change the vertex cover number.
For this reason, we restrict ourselves to flips that only remove edges and thus do not increase the vertex cover number.

The case $\chi(S) \cap \chi(V \setminus S) = \emptyset$ is then covered in Lemma \ref{la:vertex-cover-disjoint-colors}.
Here, we need to individualize up to two vertices to ensure that the number of colors increases by the desired value of $3$.

\begin{lemma}
 \label{la:vertex-cover-non-disjoint-colors}
 Let $(G,\chi)$ be a colored graph such that $|V(G)| > 1$, $G$ is connected and $\chi$ is stable with respect to $1$-WL.
 Also let $S \subseteq V(G)$ be a vertex cover of $G$ and define $I \coloneqq V(G) \setminus S$.

 Then there are colors $c,d \in \im(\chi)$ such that $|\chi^{-1}(c) \cup \chi^{-1}(d)| \geq 2$ and $vw \in E(G)$ for all distinct $v \in \chi^{-1}(c)$, $w \in \chi^{-1}(d)$, or there is some vertex $u \in V(G)$ such that, for $\chi' \coloneqq \WL{1}{G,\chi[u]}$, it holds that
 \begin{enumerate}[label=(\alph*)]
  \item\label{item:vc-non-disjoint-1} $|\chi'(S)| \geq |\chi(S)| + 2$ and $|\chi'(S) \setminus \chi'(I)| \geq |\chi(S) \setminus \chi(I)| + 1$, or
  \item\label{item:vc-non-disjoint-2} $|\chi'(S)| \geq |\chi(S)| + 1$ and $|\chi'(S) \setminus \chi'(I)| \geq |\chi(S) \setminus \chi(I)| + 2$, or
  \item\label{item:vc-non-disjoint-3} $\chi'(S) \cap \chi'(I) = \emptyset$.
 \end{enumerate}
\end{lemma}

\begin{proof}
 Suppose that there are no colors $c,d \in \im(\chi)$ such that $|\chi^{-1}(c) \cup \chi^{-1}(d)| \geq 2$ and $vw \in E(G)$ for all distinct $v \in \chi^{-1}(c)$, $w \in \chi^{-1}(d)$ (otherwise, we are already done).

 We distinguish two cases.
 First suppose there are no $v,w \in S$ such that $vw \in E(G)$.
 Then $G$ is bipartite with bipartition $(S,I)$.
 We choose an arbitrary vertex $u \in S$ and set $\chi' \coloneqq \WL{1}{G,\chi[u]}$.
 Since $G$ is connected, it follows for every $w \in V(G)$ that $w \in S$ if and only if $\dist_G(u,w)$ is even.
 Since this is detected by $1$-WL after individualizing $u$, we conclude that $\chi'(S) \cap \chi'(I) = \emptyset$, so Option \ref{item:vc-non-disjoint-3} holds.

 Otherwise, take $v,w \in S$ such that $vw \in E(G)$.
 Let $c \coloneqq \chi(v)$ and $d \coloneqq \chi(w)$.
 Also let $C \coloneqq \chi^{-1}(c)$ and $D \coloneqq \chi^{-1}(d)$.
 We distinguish the following subcases.
 \begin{itemize}[leftmargin=3ex]
  \item First suppose that $C = D$.
   We start by showing that $G[C \cap S]$ is not a complete graph.
   Suppose towards a contradiction that $G[C \cap S]$ is complete.
   Since $\chi$ is stable with respect to $1$-WL, there is a number $r \geq 0$ such that $|N(x) \cap C| = r$ for every $x \in C$.
   In fact, $r \geq |C \cap S| - 1 \geq 1$ since $G[C \cap S]$ is complete and $\{v,w\} \subseteq C \cap S$.
   Since, by assumption, $G[C]$ is not complete, there must be some vertex $y \in I \cap C$.
   Note that $N(y) \subseteq S$ and hence, $|N(y) \cap C \cap S| \geq r \geq 1$.
   But then, by degree arguments, $r \geq |C \cap S|$ since at least one vertex from $C \cap S$ is adjacent to $y$ (and all other vertices of $C \cap S$).
   So, in fact, $C \cap S \subseteq N(y)$ and $(C \cap S) \cup \{y\}$ is a clique.
   Since $G[C]$ is not complete, there must be a second vertex $z$ with $y \neq z \in I \cap C$.
   As before, we obtain that $C \cap S \subseteq N(z)$.
   But now $|N(v) \cap C| \geq |\{y,z\}| + |(C \cap S) \setminus \{v\}| > |C \cap S|$ while $N(y) \cap C \subseteq C \cap S$.
   This contradicts $|N(v) \cap C| = r = |N(y) \cap C|$.

   So $G[C \cap S]$ is not complete, which implies that there is some $u \in C \cap S$ such that $\emptyset \neq N(u) \cap C \cap S \subsetneq (C \cap S) \setminus \{u\}$.
   We set $\chi' \coloneqq \WL{1}{G,\chi[u]}$ and obtain $|\chi'(S)| \geq |\chi(S)| + 2$.
   Also, $|\chi'(S) \setminus \chi'(I)| \geq |\chi(S) \setminus \chi(I)| + 1$ since the individualized vertex $u$ is contained in $S$.
   Thus, Option \ref{item:vc-non-disjoint-1} holds.
  \item Otherwise, $C \neq D$.
   Using the same arguments as before, it follows that $G[C \cap S,D \cap S]$ is not a complete bipartite graph.
   Now, we distinguish two further subcases.
   \begin{itemize}[leftmargin=2.5ex]
    \item First suppose $|C \cap S| \geq 2$ and $|D \cap S| \geq 2$.
     Then there is some $u \in C \cap S$ such that $\emptyset \subsetneq N(u) \cap D \cap S \subsetneq D \cap S$, or there is some $u \in D \cap S$ such that $\emptyset \subsetneq N(u) \cap C \cap S \subsetneq C \cap S$.
     In both cases, we set $\chi' \coloneqq \WL{1}{G,\chi[u]}$ and obtain $|\chi'(S)| \geq |\chi(S)| + 2$.
     Also, $|\chi'(S) \setminus \chi'(I)| \geq |\chi(S) \setminus \chi(I)| + 1$ since the individualized vertex $u$ is contained in $S$ and was not a singleton before.
     Thus, Option \ref{item:vc-non-disjoint-1} holds.
    \item Otherwise $|C \cap S| = 1$ or $|D \cap S| = 1$.
     Without loss of generality, assume $|C \cap S| = 1$.
     Then $|D \cap S| \geq 2$ since $v \in C$, $w \in D$, $vw \in E(G)$ and $G[C \cap S,D \cap S]$ is not a complete bipartite graph.
     Observe that $C \cap S = \{v\}$ in this case.
     Since $G[C \cap S,D \cap S]$ is not a complete bipartite graph, we conclude that $\emptyset \subsetneq N(v) \cap D \cap S \subsetneq D \cap S$.
     Also, since $|C| \geq 2$, we conclude that $c \in \chi(S) \cap \chi(I)$.

     We choose $u \coloneqq v$ and set $\chi' \coloneqq \WL{1}{G,\chi[u]}$.
     Then $|\chi'(D \cap S)| \geq 2$ and hence, $|\chi'(S)| \geq |\chi(S)| + 1$.
     For Option \ref{item:vc-non-disjoint-2} to hold, it remains to show that $|\chi'(S) \setminus \chi'(I)| \geq |\chi(S) \setminus \chi(I)| + 2$.
     First observe that $\chi'(u) \in \chi'(S) \setminus \chi'(I)$ which means that the number of colors only appearing in $S$ already increases by one.

     For the second color, consider an arbitrary $x \in (D \cap S) \setminus N(u)$.
     Then $N(x) \cap (C \setminus \{u\}) \neq \emptyset$, and for every $y \in V(G)$ with $\chi'(x) = \chi'(y)$, it holds that $N(y) \cap (C \setminus \{u\}) \neq \emptyset$.
     Since $C \setminus \{u\} \subseteq I$, we conclude that $y \in S$ for every $y \in V(G)$ with $\chi'(x) = \chi'(y)$.
     So $\chi'(x) \in \chi'(S) \setminus \chi'(I)$.
     Also, if $\chi(x) \in \chi(S) \setminus \chi(I)$, then $\chi'(x),\chi'(w) \in \chi'(S) \setminus \chi'(I)$ and $\chi'(x) \neq \chi'(w)$, since $w \in N(u)$, but $x \notin N(u)$ and $u$ has been individualized.
     Overall, we obtain that $|\chi'(S) \setminus \chi'(I)| \geq |\chi(S) \setminus \chi(I)| + 2$.\qedhere
   \end{itemize}
 \end{itemize}
\end{proof}

Recall that we use the notation $\chi_{k,U}$ for $\WL{k}{G,\chi[U]}$, and we use $\chi_{k,u}$ for $\WL{k}{G,\chi[u]}$.

\begin{lemma}
 \label{la:vertex-cover-disjoint-colors}
 Let $(G,\chi)$ be a colored graph that is $1$-robust and suppose $|V(G)| > 1$.
 Also let $S \subseteq V(G)$ be a vertex cover of $G$ and define $I \coloneqq V(G) \setminus S$.
 Moreover, suppose that $\chi(S) \cap \chi(I) = \emptyset$.
 Then
 \begin{enumerate}[label=(\roman*)]
  \item there is a vertex $u \in V(G)$ such that $|\chi_{1,u}(S)| \geq |\chi(S)| + 2$, or
  \item there are vertices $u_1,u_2 \in V(G)$ such that $|\chi_{1,\{u_1,u_2\}}(S)| \geq |\chi(S)| + 3$.
 \end{enumerate}
\end{lemma}

\begin{proof}
 We distinguish three cases.
 \begin{itemize}[leftmargin=3ex]
  \item First suppose there are $u,w \in S$ such that $uw \in E(G)$.
   We claim that $|\chi_{1,u}(S)| \geq |\chi(S)| + 2$.
   Let $c \coloneqq \chi(u)$, $d \coloneqq \chi(w)$ and define $C \coloneqq \chi^{-1}(c)$, $D \coloneqq \chi^{-1}(d)$.
   We have that $|C|,|D| \geq 2$ since G is $1$-robust (see Lemma \ref{la:robust-no-singleton}).

   If $C \neq D$, then $D \nsubseteq N(u)$ since $G$ is flipped.
   Hence, $|\chi_{1,u}(C)| \geq 2$ and $|\chi_{1,u}(D)| \geq 2$, which implies that $|\chi_{1,u}(S)| \geq |\chi(S)| + 2$.

   Otherwise, $C = D$.
   We have $C \cap N(u) \neq \emptyset$ and $C \setminus (N[u]) \neq \emptyset$.
   So $|\chi_{1,u}(C)| \geq 3$, which implies that $|\chi_{1,u}(S)| \geq |\chi(S)| + 2$.
  \item Next, suppose there are no $u,w \in S$ such that $uw \in E(G)$, and $|\chi(S)| \geq 2$.
   Then there is some $u \in I$ such that $|\chi(N(u))| \geq 2$, i.e., $u$ is adjacent to least two color classes within $S$ (observe that $N(u) \subseteq S$).
   Pick distinct $c,d \in \chi(N(u))$ and let $C \coloneqq \chi^{-1}(c)$, $D \coloneqq \chi^{-1}(d)$.
   Since $G$ is flipped, we conclude that $C \setminus N(u) \neq \emptyset$ and $D \setminus N(u) \neq \emptyset$.
   So $|\chi_{1,u}(C)| \geq 2$ and $|\chi_{1,u}(D)| \geq 2$, which implies that $|\chi_{1,u}(S)| \geq |\chi(S)| + 2$.
  \item Finally, suppose there are no $u,w \in S$ such that $uw \in E(G)$, and $|\chi(S)| = 1$.
   Since $G$ is flipped, we get that $|N(v)| \leq \frac{1}{2}|S|$ for all $v \in I$.
   Since $G$ is connected, there must be vertices $u_1,u_2 \in I$ such that $N(u_1) \cap N(u_2) \neq \emptyset$ and $N(u_1) \nsubseteq N(u_2)$ and $N(u_2) \nsubseteq N(u_1)$.
   Then $N(u_1) \cup N(u_2) \neq S$ using that $|N(u_i)| \leq \frac{1}{2}|S|$ for both $i \in \{1,2\}$.
   It follows that $|\chi_{1,\{u_1,u_2\}}(S)| \geq 4 = |\chi(S)| + 3$.\qedhere
 \end{itemize}
\end{proof}

With Lemmas \ref{la:vertex-cover-non-disjoint-colors} and \ref{la:vertex-cover-disjoint-colors} at our disposal, we can now bound the $1$-WL depth in terms of the vertex cover number.

\begin{theorem}
 \label{thm:wl-depth-vertex-cover}
  Let $(G,\chi)$ be a colored graph with vertex cover number $r$. Then
  \[\WLd{1}{G,\chi} \leq \frac{2}{3} \cdot r + 1.\]
\end{theorem}

\begin{proof}
 Let $(G,\chi)$ be a colored graph and let $S$ be a vertex cover of $G$.
 Also, define $I \coloneqq V(G) \setminus S$.
 We prove that
 \begin{equation}
  \label{eq:vertex-cover-bound}
  \WLd{1}{G,\chi} \leq \delta_{\chi,S} + \frac{1}{3}\Big(2 |S| - |\chi(S)| - |\chi(S) \setminus \chi(I)|\Big)
 \end{equation}
 where
 \[\delta_{\chi,S} \coloneqq \begin{cases}
                              1 &\text{if }\chi(S) \cap \chi(I) \neq \emptyset,\\
                              0 &\text{otherwise.}
                             \end{cases}
 \]
 First observe that this implies the desired bound since
 \[\delta_{\chi,S} + \frac{1}{3}\Big(2 |S| - |\chi(S)| - |\chi(S) \setminus \chi(I)|\Big) \leq 1 + \frac{2}{3}|S|\]
 and we can choose $S$ to be a vertex cover of minimal size.

 We prove \eqref{eq:vertex-cover-bound} by induction on the tuple $(|V(G)|,|\im(\chi)|,|E(G)|)$.
 Consider the set $M \coloneqq \{(n,\ell,m) \in \ZZ_{\geq 0}^3 \mid \ell \leq n\}$ and observe that $(|V(G)|,|\im(\chi)|,|E(G)|) \in M$.
 For the induction, we define a linear order $\prec$ on $M$ via $(n,\ell,m) \prec (n',\ell',m')$ if $n < n'$, or $n = n'$ and $\ell > \ell'$, or $n = n'$ and $\ell = \ell'$ and $m < m'$.
 Note that we use the inverse order on the second component, i.e., for $\ell \neq \ell'$, we have $(n,\ell,m) \prec (n,\ell',m')$ if $\ell > \ell'$.
 Still, since $\ell \leq n$ for every $(n,\ell,m) \in M$, there are no infinite decreasing chains in $M$.

 For the base case, suppose that $|V(G)| = 1$.
 Then $\WLd{1}{G,\chi} = 0 \leq \delta_{\chi,S} + \frac{1}{3}(2 |S| - |\chi(S)| - |\chi(S) \setminus \chi(I)|)$.

 So suppose $|V(G)| > 1$.
 We distinguish several cases.
 \begin{itemize}[leftmargin=3ex]
  \item First assume that $\chi$ is not stable with respect to $1$-WL and let $\chi' \coloneqq \WL{1}{G,\chi}$.
   Note that $|\im(\chi')| > |\im(\chi)|$ and thus, $(|V(G)|,|\im(\chi')|,|E(G)|) \prec (|V(G)|,|\im(\chi)|,|E(G)|)$.
   Hence, by the induction hypothesis, we get
   \[\WLd{1}{G,\chi'} \leq \delta_{\chi',S} + \frac{1}{3}\Big(2 |S| - |\chi'(S)| - |\chi'(S) \setminus \chi'(I)|\Big).\]
   Together with Lemma \ref{la:wl-depth-bound}\ref{item:wl-depth-bound-1}, we get that
   \begin{align*}
    \WLd{1}{G,\chi} \leq \WLd{1}{G,\chi'} &\leq \delta_{\chi',S} + \frac{1}{3}\Big(2 |S| - |\chi'(S)| - |\chi'(S) \setminus \chi'(I)|\Big)\\
                                          &\leq \delta_{\chi,S} + \frac{1}{3}\Big(2 |S| - |\chi(S)| - |\chi(S) \setminus \chi(I)|\Big)
   \end{align*}
   where the last inequality holds because $\chi' \preceq \chi$.
  \item Next assume that $G$ is not connected.
   Let $A_1,\dots,A_\ell$ denote the vertex sets of the connected components of $G$.
   For $i \in [\ell]$, let $(G_i,\chi_i) \coloneqq (G[A_i],\chi|_{A_i})$.
   Observe that $S_i \coloneqq S \cap A_i$ is a vertex cover of $G_i$ and define $I_i \coloneqq A_i \setminus S_i$.

   Note that $(|V(G_i)|,|\im(\chi_i)|,|E(G_i)|) \prec (|V(G)|,|\im(\chi)|,|E(G)|)$ for all $i \in [\ell]$ (because $|V(G_i)| < |V(G)|$).
   Hence, by the induction hypothesis, for every $i \in [\ell]$, we get
   \[\WLd{1}{G_i,\chi_i} \leq \delta_{\chi_i,S_i} + \frac{1}{3}\Big(2 |S_i| - |\chi_i(S_i)| - |\chi_i(S_i) \setminus \chi_i(I_i)|\Big) \eqqcolon d_i.\]
   For every $i \in [\ell]$, it holds that
   \begin{enumerate}
    \item $\delta_{\chi_i,S_i} \leq \delta_{\chi,S}$,
    \item $|S_i| - |\chi_i(S_i)| \leq |S| - |\chi(S)|$, and
    \item $|S_i| - |\chi_i(S_i)\setminus \chi_i(I_i)| \leq |S_i| - |\chi(S_i)\setminus \chi(I)| \leq |S| - |\chi(S)\setminus \chi(I)|$.
   \end{enumerate}
   Together, we get that
   \[d_i \leq \delta_{\chi,S} + \frac{1}{3}\Big(2 |S| - |\chi(S)| - |\chi(S) \setminus \chi(I)|\Big)\]
   for all $i \in [\ell]$.
   In combination with Lemma \ref{la:wl-depth-bound}\ref{item:wl-depth-bound-3} we conclude that
   \begin{align*}
    \WLd{1}{G,\chi} \leq \max_{i \in [\ell]}\WLd{k}{G_i,\chi_i} \leq \delta_{\chi,S} + \frac{1}{3}\Big(2 |S| - |\chi(S)| - |\chi(S) \setminus \chi(I)|\Big).
   \end{align*}
  \item Now suppose that $\chi$ is stable with respect to $1$-WL, $G$ is connected, and $\chi(S) \cap \chi(I) \neq \emptyset$.
   We consider the following subcases.
   \begin{itemize}[leftmargin=2.5ex]
    \item First suppose that there are $c,d \in C \coloneqq \im(\chi)$ such that $|\chi^{-1}(c) \cup \chi^{-1}(d)| \geq 2$ and $vw \in E(G)$ for all distinct $v \in \chi^{-1}(c), w \in \chi^{-1}(d)$.
     Consider the flip function $f\colon C \times C \to \{0,1\}$ defined via $f(c,d) = f(d,c) \coloneqq 1$ and $f(c',d') \coloneqq 0$ for all $c',d' \in C$ such that $\{c',d'\} \neq \{c,d\}$.
     Let $(G',\chi') \coloneqq \flipf_f(G,\chi)$.
     Note that $V(G) = V(G')$, $\chi' = \chi$, $E(G') \subsetneq E(G)$.
     So $(|V(G')|,|\im(\chi')|,|E(G')|) \prec (|V(G)|,|\im(\chi)|,|E(G)|)$ and $S$ is also a vertex cover of $G'$.
     Hence, by the induction hypothesis, we get
     \[\WLd{1}{G',\chi'} \leq \delta_{\chi',S} + \frac{1}{3}\Big(2 |S| - 2 |\chi'(S)| - |\chi'(S) \setminus \chi'(I)|\Big).\]
     Together with Lemma \ref{la:wl-depth-bound}\ref{item:wl-depth-bound-2} we get that
     \begin{align*}
      \WLd{1}{G,\chi} \leq \WLd{1}{G',\chi'} &\leq \delta_{\chi',S} + \frac{1}{3}\Big(2 |S| - |\chi'(S)| - |\chi'(S) \setminus \chi'(I)|\Big)\\
                                             &= \delta_{\chi,S} + \frac{1}{3}\Big(2 |S| - |\chi(S)| - |\chi(S) \setminus \chi(I)|\Big)
     \end{align*}
     where the last equality holds because $\chi' = \chi$.
    \item Otherwise, by Lemma \ref{la:vertex-cover-non-disjoint-colors}, there is a vertex $u \in V(G)$ such that, for $\chi' \coloneqq \WL{1}{G,\chi[u]}$, one of the Options \ref{item:vc-non-disjoint-1}-\ref{item:vc-non-disjoint-3} is satisfied.
     In all three cases, we get that $\chi' \preceq \chi$ and $|\im(\chi')| > |\im(\chi)|$.
     Thus $(|V(G)|,|\im(\chi')|,|E(G)|) \prec (|V(G)|,|\im(\chi)|,|E(G)|)$.
     Hence, by the induction hypothesis, we get
     \[\WLd{1}{G',\chi'} \leq \delta_{\chi',S} + \frac{1}{3}\Big(2 |S| - |\chi'(S)| - |\chi'(S) \setminus \chi'(I)|\Big).\]
     Using Lemma \ref{la:wl-depth-bound}, Items \ref{item:wl-depth-bound-1} and \ref{item:wl-depth-bound-4} we conclude that
     \[\WLd{1}{G,\chi} \leq 1 + \WLd{1}{G,\chi[u]} \leq 1 + \WLd{1}{G,\chi'}.\]
     In order to bound the term $1 + \WLd{1}{G,\chi'}$, we need to consider which of the Options \ref{item:vc-non-disjoint-1}-\ref{item:vc-non-disjoint-3} is satisfied.
     If Option \ref{item:vc-non-disjoint-1} is satisfied, then
     \begin{align*}
      1 + \WLd{1}{G,\chi'} &\leq 1 + \delta_{\chi',S} + \frac{1}{3}\Big(2 |S| - |\chi'(S)| - |\chi'(S) \setminus \chi'(I)|\Big)\\
                           &\leq 1 + \delta_{\chi,S} + \frac{1}{3}\Big(2 |S| - (|\chi(S)| + 2) - (|\chi(S) \setminus \chi(I)| + 1)\Big)\\
                           &= \delta_{\chi,S} + \frac{1}{3}\Big(2 |S| - |\chi(S)| - |\chi(S) \setminus \chi(I)|\Big)
     \end{align*}
     where the last inequality holds because $\chi' \preceq \chi$ and $|\chi'(S)| \geq |\chi(S)| + 2$ and $|\chi'(S) \setminus \chi'(I)| \geq |\chi(S) \setminus \chi(I)| + 1$.
     If Option \ref{item:vc-non-disjoint-2} is satisfied, then
     \begin{align*}
      1 + \WLd{1}{G,\chi'} &\leq 1 + \delta_{\chi',S} + \frac{1}{3}\Big(2 |S| - |\chi'(S)| - |\chi'(S) \setminus \chi'(I)|\Big)\\
                           &\leq 1 + \delta_{\chi,S} + \frac{1}{3}\Big(2 |S| - (|\chi(S)| + 1) - (|\chi(S) \setminus \chi(I)| + 2)\Big)\\
                           &= \delta_{\chi,S} + \frac{1}{3}\Big(2 |S| - |\chi(S)| - |\chi(S) \setminus \chi(I)|\Big)
     \end{align*}
     where the last inequality holds because $\chi' \preceq \chi$, $|\chi'(S)| \geq |\chi(S)| + 1$ and $|\chi'(S) \setminus \chi'(I)| \geq |\chi(S) \setminus \chi(I)| + 2$.
     Finally, if Option \ref{item:vc-non-disjoint-3} is satisfied, then $\delta_{\chi',S} = 0$.
     Also $\delta_{\chi,S} = 1$ by assumption.
     So together
     \begin{align*}
      1 + \WLd{1}{G,\chi'} &\leq 1 + \delta_{\chi',S} + \frac{1}{3}\Big(2 |S| - |\chi'(S)| - |\chi'(S) \setminus \chi'(I)|\Big)\\
                           &= \delta_{\chi,S} + \frac{1}{3}\Big(2 |S| - |\chi'(S)| - |\chi'(S) \setminus \chi'(I)|\Big)\\
                           &\leq \delta_{\chi,S} + \frac{1}{3}\Big(2 |S| - |\chi(S)| - |\chi(S) \setminus \chi(I)|\Big)
     \end{align*}
     where the last equality holds because $\chi' \preceq \chi$.

     In all three cases, we obtain
     \[\WLd{1}{G,\chi} \leq \delta_{\chi,S} + \frac{1}{3}\Big(2 |S| - |\chi(S)| - |\chi(S) \setminus \chi(I)|\Big)\]
     as desired.
    \end{itemize}
  \item Finally, suppose that $\chi$ is stable with respect to $1$-WL, $G$ is connected, and $\chi(S) \cap \chi(I) = \emptyset$.
   We again need to distinguish some subcases.
   \begin{itemize}[leftmargin=2.5ex]
    \item First suppose that $(G,\chi)$ is not $1$-robust, i.e., $(G,\chi)$ is not flipped.
     Let $(G',\chi') \coloneqq \flip(G,\chi)$.
     Note that $V(G) = V(G')$, $\chi' = \chi$, $|E(G')| < |E(G)|$.
     So we conclude that $(|V(G')|,|\im(\chi')|,|E(G')|) \prec (|V(G)|,|\im(\chi)|,|E(G)|)$.
     Also, since $\chi(S) \cap \chi(I) = \emptyset$, no edges between vertices in $I$ are added in the flip.
     So $S$ is still a vertex cover of $G'$.
     Hence, by the induction hypothesis, we get
     \[\WLd{1}{G',\chi'} \leq \delta_{\chi',S} + \frac{1}{3}\Big(2 |S| - |\chi'(S)| - |\chi'(S) \setminus \chi'(I)|\Big).\]
     Together with Lemma \ref{la:wl-depth-bound}\ref{item:wl-depth-bound-2}, we get that
     \begin{align*}
      \WLd{1}{G,\chi} \leq \WLd{1}{G',\chi'} &\leq \delta_{\chi',S} + \frac{1}{3}\Big(2 |S| - |\chi'(S)| - |\chi'(S) \setminus \chi'(I)|\Big)\\
                                             &= \delta_{\chi,S} + \frac{1}{3}\Big(2 |S| - |\chi(S)| - |\chi(S) \setminus \chi(I)|\Big)
     \end{align*}
     where the last equality holds because $\chi' = \chi$.
    \item Next, suppose there $u \in V(G)$ such that $|\chi'(S)| \geq |\chi(S)| + 2$ where $\chi' \coloneqq \WL{1}{G,\chi[u]}$.
     Since $\chi(S) \cap \chi(I) = \emptyset$ by assumption, we also get $|\chi'(S) \setminus \chi'(I)| \geq |\chi(S) \setminus \chi(I)| + 2 \geq |\chi(S) \setminus \chi(I)| + 1$.
     So this case is identical to the corresponding subcase above.
    \item Otherwise, by Lemma \ref{la:vertex-cover-disjoint-colors}, there are $u_1,u_2 \in V(G)$ such that $|\chi'(S)| \geq |\chi(S)| + 3$ where $\chi' \coloneqq \WL{1}{G,\chi[u_1,u_2]}$.
     We get that $\chi' \preceq \chi$ and $|\im(\chi')| > |\im(\chi)|$.
     So we get $(|V(G)|,|\im(\chi')|,|E(G)|) \prec (|V(G)|,|\im(\chi)|,|E(G)|)$.
     Hence, by the induction hypothesis, we get
     \[\WLd{1}{G,\chi'} \leq \delta_{\chi',S} + \frac{1}{3}\Big(2 |S| - |\chi'(S)| - |\chi'(S) \setminus \chi'(I)|\Big).\]
     Applying Lemma \ref{la:wl-depth-bound}\ref{item:wl-depth-bound-4} twice and after that Lemma \ref{la:wl-depth-bound}\ref{item:wl-depth-bound-1} we obtain
     \begin{equation*}
      \WLd{1}{G,\chi} \leq 2 + \WLd{1}{G,\chi'}.
     \end{equation*}
     Note that $\chi(S) \cap \chi(I) = \emptyset$ which means that $\delta_{\chi,S} = \delta_{\chi',S} = 0$.
     It also implies that $\chi(S) \setminus \chi(I) = \chi(S)$ and $\chi'(S) \setminus \chi'(I) = \chi'(S)$.
     Altogether, we have
     \begin{align*}
      \WLd{1}{G,\chi} &\leq 2 + \delta_{\chi',S} + \frac{1}{3}\Big(2 |S| - |\chi'(S)| - |\chi'(S) \setminus \chi'(I)|\Big)\\
                      &\leq 2 + \delta_{\chi,S} + \frac{1}{3}\Big(2 |S| - (|\chi(S)| + 3) - (|\chi(S) \setminus \chi(I)| + 3)\Big)\\
                      &= 2 + \delta_{\chi,S} + \frac{1}{3}\Big(2 |S| - |\chi(S)| - |\chi(S) \setminus \chi(I)| - 6\Big)\\
                      &\leq \delta_{\chi,S} + \frac{1}{3}\Big(2 |S| - |\chi(S)| - |\chi(S) \setminus \chi(I)|\Big).\qedhere
     \end{align*}
   \end{itemize}
 \end{itemize}
\end{proof}

\begin{corollary}
 Let $(G,\chi)$ be a colored graph with vertex cover number $r$. Then $(G,\chi)$ has WL dimension at most $\frac{2}{3} \cdot r + 3$.
\end{corollary}

\begin{proof}
 This follows from Lemma \ref{la:wl-depth-to-dimension} and Theorem \ref{thm:wl-depth-vertex-cover}.
\end{proof}

\section{Advanced Bounds Using 2-WL}\label{sec:advanced}

In this section, we improve on the bound on the WL depth of $n$-vertex graphs stated in Corollary \ref{cor:wl-depth-robust} and prove Theorem \ref{thm:main}.
Towards this end, we need to rely on $2$-WL.
More precisely, we analyze the $2$-WL depth of $n$-vertex graphs.
The following theorem is the main result of this section.

\begin{theorem}
 \label{thm:wl-depth-general}
 Let $(G,\chi)$ be a colored graph and let $n \coloneqq |V(G)|$. Then $\WLd{2}{G,\chi} \leq \frac{n}{4} + o(n)$.
\end{theorem}

Together with Lemma \ref{la:wl-depth-to-dimension}, this immediately gives the following corollary.

\begin{corollary}
 Let $(G,\chi)$ be a colored graph and let $n \coloneqq |V(G)|$. Then the WL dimension of $(G,\chi)$ is at most $\frac{n}{4} + o(n)$.
\end{corollary}

In particular, we obtain Theorem \ref{thm:main}.
Also, together with Theorem \ref{thm:eq-wl-ck}, we obtain Corollary \ref{cor:main}.

On a high level, the proof of Theorem \ref{thm:wl-depth-general}, which covers the rest of this section, again follows a similar strategy as the one for Corollary \ref{cor:wl-depth-robust}.
However, to obtain the improved bound, the details become significantly more intricate and we have to distinguish several cases.

Let $(G,\chi)$ be a vertex-colored graph and suppose that $\chi$ is stable with respect to $2$-WL.
We define
\[\mu_{G,\chi}(c_1,c_2) \coloneqq |\{\WL{2}{G,\chi}(v_1,v_2) \mid v_1 \in \chi^{-1}(c_1),v_2 \in \chi^{-1}(c_2), v_1 \neq v_2\}|\]
for all $c_1,c_2 \in \im(\chi)$.
Also, similarly as in Lemma \ref{la:color-class-increase-general-robust}, as a lower bound on the progress, i.e., how many new colors we get from a $2$-WL refinement of a coloring after a vertex individualization, we set
\[\xi_{G,\chi}(c_1) \coloneqq \sum_{c_2 \in \im(\chi)}(\mu_{G,\chi}(c_1,c_2) - 1).\]
Note that $\xi_{G,\chi}(c_1) \geq \deg_{G[[\chi]]}(c_1)$ for all colors $c_1 \in \im(\chi)$.

As before, we can restrict our attention to $2$-robust graphs $(G,\chi)$ for which $|V(G)| > 1$.
Towards this end, a vertex-colored graph $(G,\chi)$ is called \emph{nice} if it is connected, flipped, $|V(G)| > 1$, and $\chi$ is stable with respect to $2$-WL.

The next lemma, which is similar in nature to Lemma \ref{la:color-class-increase-general-robust}, forms the starting point of our analysis.

\begin{lemma}
 \label{la:split-many-pair-colors}
 Let $(G,\chi)$ be a nice graph.
 Let $c \in \im(\chi)$ and set $\xi \coloneqq \xi_{G,\chi}(c) + 1$.
 Then $|\im(\chi_{2,u})| \geq |\im(\chi)| + \xi$ for every $u \in \chi^{-1}(c)$.
\end{lemma}

\begin{proof}
 Clearly, $\chi_{2,u} \preceq \chi$.
 So
 \begin{align*}
  |\im(\chi_{2,u})| = \sum_{c' \in \im(\chi)} |\chi_{2,u}(\chi^{-1}(c'))| \geq 1 + \sum_{c' \in \im(\chi)} \mu_{G,\chi(c,c')} = |\im(\chi)| + \xi,
 \end{align*}
 where the inequality follows from $\chi$ being stable with respect to $2$-WL.
\end{proof}

If $\xi_{G,\chi}(c) \geq 3$, then $\xi \geq 4$ in the statement of the lemma.
Thus, individualizing any vertex $u$ of color $c$ yields an increase in the number of color classes by at least $4$, which is sufficient to show $\WLd{2}{G,\chi} \leq \frac{n}{4} + o(n)$.
In the following, we thus consider nice graphs $(G,\chi)$ such that $\xi_{G,\chi}(c) \leq 2$ holds for all colors $c \in \im(\chi)$.
In particular, this means that $G[[\chi]]$ has maximum degree $2$.
To cover the remaining cases, we rely on several tools from \cite{Babai80,Babai81}.

Let $(G,\chi)$ be a colored graph and let $\lambda \coloneqq \WL{2}{G,\chi}$.
For every $c \in \im(\lambda)$, we define the \emph{constituent graph of $c$} to be the (undirected) graph $F_c$ with vertex set $V(F_c) \coloneqq V(G)$ and edge set
\[E(F_c) \coloneqq \{vw \mid \lambda(v,w) = c\}.\]

\begin{theorem}[{\cite[Theorem 2.1]{Babai81}}]
 \label{thm:uniprimitive-individualization}
 Let $(G,\chi)$ be a colored graph and suppose that $\lambda \coloneqq \WL{2}{G,\chi}$.
 Also assume that
 \begin{enumerate}[label=(\Alph*)]
  \item $c_0 \coloneqq \lambda(v,v) = \lambda(w,w)$ for all $v,w \in V(G)$,
  \item $|\im(\lambda)| \geq 3$, and
  \item $F_c$ is connected for every color $c \in \im(\lambda) \setminus \{c_0\}$.
 \end{enumerate}
 Then there is a set $U \subseteq V(G)$ of size
 \[|U| < 4\sqrt{n}\log n\]
 such that $\chi_{2,U}$ is discrete.
\end{theorem}

For two sets $A$ and $B$, we let $A \triangle{} B \coloneqq (A \setminus B) \cup (B \setminus A)$ denote the symmetric difference of $A$ and $B$.

\begin{lemma}[\cite{Babai80}]
 \label{la:distinguishing-set-from-fractional-cover}
 Let $(G,\chi)$ be a graph and let $D \subseteq V(G)$ such that $|N(v) \triangle N(w)| \geq \ell$ for all distinct $v,w \in D$.
 Then there is a set $U \subseteq V(G)$ such that $|U| \leq \frac{n}{\ell}(1 + 2\log n)$ and $\chi_U(v,v) \neq \chi_U(w,w)$ for all distinct $v,w \in D$.
\end{lemma}

A \emph{hypergraph} is a pair $\CH = (V,\CE)$ where $V$ is a non-empty finite set of vertices and $\CE \subseteq 2^V$ is the set of \emph{hyperedges}.
A hypergraph $\CH = (V,\CE)$ is \emph{non-empty} if $\CE \neq \emptyset$.
Also, $\CH$ is \emph{$r$-uniform} if $|E| = r$ for all $E \in \CE$.
Moreover, $\CH$ is \emph{$k$-regular} if for every $v \in V$, it holds that $|\{E \in \CE \mid v \in E\}| = k$.
We say $\CH$ is \emph{regular} if $\CH$ is $k$-regular for some $k \geq 0$.

\begin{lemma}[{\cite[Lemma 3.4]{Babai80}}]
 \label{la:hypergraph-bounds}
 Suppose that $1 \leq d < r < n$.
 Let $\CH = (V,\CE)$ be a non-empty regular $r$-uniform hypergraph such that $|E \cap E'| \geq d$ for all $E,E' \in \CE$.
 Then $r^2 > nd$.
\end{lemma}

Also, our analysis relies on the following auxiliary lemma.

\begin{lemma}
 \label{la:bipartite-sr-individualization}
 Let $(G,\chi)$ be a colored graph and suppose that $\lambda \coloneqq \WL{2}{G,\chi}$.
 Also assume that there is a partition $(V_1,V_2)$ of $V(G)$ such that
 \begin{enumerate}[label=(\Alph*)]
  \item\label{item:bipartite-sr-individualization-1} $|V_1| \geq 4$ and $|V_2| \geq 4$,
  \item\label{item:bipartite-sr-individualization-2} $2 \leq |N_G(v_1) \cap V_2| \leq |V_2| - 2$ for all $v_1 \in V_1$, and
  \item\label{item:bipartite-sr-individualization-3} $|\lambda(V_1,V_1)| = 2$, $|\lambda(V_2,V_2)| = 2$ and $\lambda(V_1,V_1) \cap \lambda(V_2,V_2) = \emptyset$.
 \end{enumerate}
 Then there is a set $U \subseteq V(G)$ of size
 \[|U| < 6\sqrt{n}\log n\]
 such that $\chi_{2,U}$ is discrete.
\end{lemma}

\begin{proof}
 Combining \ref{item:bipartite-sr-individualization-1} and \ref{item:bipartite-sr-individualization-3}, we conclude that $k_1 \coloneqq \deg(v_1) = \deg(v_1')$ for all $v_1,v_1' \in V_1$.
 Similarly, $k_2 \coloneqq \deg(v_2) = \deg(v_2')$ for all $v_2,v_2' \in V_2$.
 We also have $\lambda(v_1,v_1') = \lambda(v_1'',v_1''')$ for all $v_1,v_1',v_1'',v_1''' \in V_1$ such that $v_1 \neq v_1'$ and $v_1'' \neq v_1'''$.
 In particular, $\mu_1 \coloneqq |N(v_1) \cap N(v_1')| = |N(v_1'') \cap N(v_1''')|$ for all $v_1,v_1',v_1'',v_1''' \in V_1$ such that $v_1 \neq v_1'$ and $v_1'' \neq v_1'''$.
 Similarly, $\lambda(v_2,v_2') = \lambda(v_2'',v_2''')$ for all $v_2,v_2'v_2'',v_2''' \in V_2$ such that $v_2 \neq v_2'$ and $v_2'' \neq v_2'''$.
 We get that $2 \leq k_2 \leq |V_1| - 2$ and $\mu_2 \coloneqq |N(v_2) \cap N(v_2')| = |N(v_2'') \cap N(v_2''')|$ for all $v_2,v_2',v_2'',v_2''' \in V_1$ such that $v_2 \neq v_2'$ and $v_2'' \neq v_2'''$.

 Without loss of generality, assume that $|V_1| \leq |V_2|$.
 By transitioning to the complement of $G$, we may also suppose that $k_2 \leq \frac{1}{2} \cdot |V_1|$.
 Also note that $\mu_1 \geq 1$ by Condition \ref{item:bipartite-sr-individualization-2}.
 Moreover, $k_1 \cdot k_2 \geq |V_2| - 1$ and $|V_1| \cdot k_1 = |V_2| \cdot k_2$.

 \begin{claim}
  $k_1^2 > \mu_1 |V_2|$ and $k_2^2 > \mu_2 |V_1|$.
 \end{claim}
 \begin{claimproof}
  By symmetry, it suffices to prove the first statement.
  We define a hypergraph $\CH = (V_2,\CE)$ where
  \[\CE \coloneqq \{N(v_1) \mid v_1 \in V_1\}.\]
  Then $\CH$ is a regular and $k_1$-uniform hypergraph such that $|E \cap E'| \geq \mu_1$ for all $E,E' \in \CE$.
  So $k_1^2 > \mu_1 |V_2|$ by Lemma \ref{la:hypergraph-bounds}.
 \end{claimproof}

 Since $k_2 \leq \frac{1}{2}|V_1|$, we obtain that $\frac{1}{2}k_2 > \mu_2$.
 By the same argument, we have $\frac{1}{2}k_1 > \mu_1$.

 Let $v_1,v_1' \in V_1$.
 Then $|N(v_1) \triangle N(v_1')| = 2(k_1 - \mu_1) > k_1$.
 So, by Lemma \ref{la:distinguishing-set-from-fractional-cover}, there is a set $U \subseteq V(G)$ such that $|U| \leq \frac{n}{k_1}(1 + 2\log n)$ and $\chi_U(v_1,v_1) \neq \chi_U(v_1',v_1')$ for all distinct $v_1,v_1' \in V_1$.
 This implies that $\chi_U$ is discrete.
 Also, $k_1^2 \geq |V_2| - 1 \geq \frac{1}{2}(n-1) \geq \frac{1}{4}n$ which implies that
 \[\frac{n}{k_1}(1 + 2\log n) \leq 2\sqrt{n}(1 + 2\log n) \leq 6\sqrt{n}\log n.\qedhere\]
\end{proof}

Recall that, with Lemma \ref{la:split-many-pair-colors} at hand, we can restrict our attention to nice graphs $(G,\chi)$ such that $\xi_{G,\chi}(c) \leq 2$ for all colors $c \in \im(\chi)$.
Also, with Lemma \ref{la:remove-twins} in mind, we may assume that $(G,\chi)$ is twin-free.
The next lemma covers the first case when $|\im(\chi)| \geq 4$.

\begin{lemma}
 \label{la:split-with-4-vertex-colors}
 Let $(G,\chi)$ be a nice and twin-free graph.
 Also suppose $\xi_{G,\chi}(c) \leq 2$ for all colors $c \in \im(\chi)$, and $|\im(\chi)| \geq 4$.

 Then there is a set $U \subseteq V(G)$ of size
 \[|U| < 6\sqrt{n}\log n\]
 such that $\chi_{2,U}$ is discrete.
\end{lemma}

\begin{proof}
 Since $\xi_{G,\chi}(c) \leq 2$ for all $c \in \im(\chi)$, we have that $G[[\chi]]$ has maximum degree $2$.
 Also, $G[[\chi]]$ is connected since $G$ is connected.
 So there are colors $c_1,c_2,c_3,c_4 \in \im(\chi)$ such that $(c_1,c_2), (c_2,c_3), (c_3,c_4) \in E(G[[\chi]])$.
 Since $(G,\chi)$ is flipped and $\xi_{G,\chi}(c) \leq 2$ for all $c \in \im(\chi)$, we conclude that
 \[\mu_{G,\chi}(c_1,c_2) = \mu_{G,\chi}(c_2,c_3) = \mu_{G,\chi}(c_3,c_4) = 2.\]
 Note that this implies $\mu_{G,\chi}(c_2,c_2) = \mu_{G,\chi}(c_3,c_3) = 1$ since $\xi_{G,\chi}(c_2) \leq 2$ and $\xi_{G,\chi}(c_3) \leq 2$.
 Let $V_2 \coloneqq \chi^{-1}(c_2)$ and $V_3 \coloneqq \chi^{-1}(c_3)$.

 We verify that $(G',\chi') \coloneqq (G[V_2 \cup V_3],\chi|_{V_2 \cup V_3})$ satisfies the prerequisites of Lemma \ref{la:bipartite-sr-individualization}.
 Let $\lambda \coloneqq \WL{2}{G,\chi}$ and $\lambda' \coloneqq \WL{2}{G',\chi'}$.
 Since $\mu_{G,\chi}(c_2,c_2) = \mu_{G,\chi}(c_3,c_3) = 1$, we get that $|\lambda(V_2,V_2)| = 2$, $|\lambda(V_3,V_3)| = 2$.
 Thus, $|\lambda'(V_2,V_2)| = 2$, $|\lambda'(V_3,V_3)| = 2$.
 Also it clearly holds that $\lambda'(V_2,V_2) \cap \lambda'(V_3,V_3) = \emptyset$.
 So Condition \ref{item:bipartite-sr-individualization-3} is satisfied.

 For Condition \ref{item:bipartite-sr-individualization-2}, observe that $G$ cannot induce a matching between $V_2$ and $V_3$, since otherwise we obtain $\mu_{G,\chi}(c_2,c_4) = \mu_{G,\chi}(c_3,c_4) = 2$.
 But this is not possible since $\xi_{G,\chi}(c_2) \leq 2$.
 Using additionally that $(G,\chi)$ is flipped, this implies that Condition \ref{item:bipartite-sr-individualization-2} is satisfied.

 Altogether, this also implies that $|V_2| \geq 4$ and $|V_3| \geq 4$, i.e., Condition \ref{item:bipartite-sr-individualization-1} is satisfied.

 So we can apply Lemma \ref{la:bipartite-sr-individualization} to obtain a set $U \subseteq V(G)$ of size
 \[|U| < 6\sqrt{n}\log n\]
 such that $\chi_{2,U}$ is discrete on the sets $V_2$ and $V_3$.
 However, then $\chi_{2,U}$ is discrete (on the entire vertex set of $G$) since $G$ is twin-free and $\xi_{G,\chi}(c) \leq 2$ for all colors $c \in \im(\chi)$.
\end{proof}

To cover the remaining cases, we additionally require the following auxiliary lemma, which is a variant of Lemma \ref{la:bipartite-sr-individualization}.

\begin{lemma}
 \label{la:bipartite-individualization-weak}
 Let $(G,\chi)$ be a colored and twin-free graph and suppose that $\lambda \coloneqq \WL{2}{G,\chi}$.
 Also assume that there is a partition $(V_1,V_2)$ of $V(G)$ such that
 \begin{enumerate}[label=(\Alph*)]
  \item\label{item:bipartite-individualization-weak-1} $|V_1| \geq 4$ and $|V_2| \geq 4$,
  \item\label{item:bipartite-individualization-weak-2} $2 \leq |N_G(v_2) \cap V_1| \leq |V_1| - 2$ for all $v_2 \in V_2$,
  \item\label{item:bipartite-individualization-weak-3} $|\lambda(V_1,V_1)| = 2$ and $\lambda(V_1,V_1) \cap \lambda(V_2,V_2) = \emptyset$, and
  \item\label{item:bipartite-individualization-weak-4} $\lambda(v_2,v_2) = \lambda(v_2',v_2')$ for all $v_2,v_2' \in V_2$ and $|\lambda(V_2,V_2)| \leq 3$.
 \end{enumerate}
 Then there is a set $U \subseteq V(G)$ of size
 \[|U| \leq 6n^{3/4}\log n\]
 such that $\chi_{2,U}$ is discrete.
\end{lemma}

\begin{proof}
 We start by showing the following claim.
 \begin{claim}
  \label{claim:individualize-one-side}
  Let $U \subseteq V(G)$ such that
  \begin{enumerate}
   \item $\chi_{2,U}(v_1) \neq \chi_{2,U}(v_1')$ for all distinct $v_1,v_1' \in V_1$, or
   \item $\chi_{2,U}(v_2) \neq \chi_{2,U}(v_2')$ for all distinct $v_2,v_2' \in V_2$.
  \end{enumerate}
  Then $\chi_{2,U}$ is discrete.
 \end{claim}
 \begin{claimproof}
  First note that $\chi_{2,U}(v_1) \neq \chi_{2,U}(v_2)$ for all $v_1 \in V_1$ and $v_2 \in V_2$ by Condition \ref{item:bipartite-individualization-weak-3}.

  Suppose the first option is satisfied, i.e., $\chi_{2,U}(v_1) \neq \chi_{2,U}(v_1')$ for all distinct $v_1,v_1' \in V_1$.
  To show the claim, it suffices to argue that $N(v_2) \cap V_1 \neq N(v_2') \cap V_1$ for all distinct $v_2,v_2' \in V_2$.
  Suppose towards a contradiction that there are distinct $v_2,v_2' \in V_2$ such that $N(v_2) \cap V_1 = N(v_2') \cap V_1$.
  Consider the corresponding equivalence relation $\sim$ defined on $V_2$ via $v_2 \sim v_2'$ if $N(v_2) \cap V_1 = N(v_2') \cap V_1$.
  Let $A_1,\dots,A_\ell$ denote the equivalence classes.
  Note that $\ell > 1$ since $G$ is connected and flipped.
  Since pairs of equivalent vertices receive a different color under $\lambda$ than non-equivalent vertices and $|\lambda(V_2,V_2)| \leq 3$ (see Condition \ref{item:bipartite-individualization-weak-4}), we conclude that $G[V_2]$ is either empty or a disjoint union of cliques corresponding to $A_1,\dots,A_\ell$.
  In both cases, distinct equivalent vertices are twins in $G$ which is a contradicts $G$ being twin-free.

  The proof for the second option is analogous.
 \end{claimproof}

 We distinguish two cases.
 First suppose that $|V_1| \leq n^{3/4}$ or $|V_2| \leq n^{3/4}$.
 Then we set $U \coloneqq V_1$ and $U \coloneqq V_2$, respectively.
 Using Claim \ref{claim:individualize-one-side} it follows that $\chi_{2,U}$ is discrete.

 So suppose that $|V_1| \geq n^{3/4}$ and $|V_2| \geq n^{3/4}$.
 Combining \ref{item:bipartite-individualization-weak-1} and \ref{item:bipartite-individualization-weak-3}, we conclude that $k_1 \coloneqq \deg(v_1) = \deg(v_1')$ for all $v_1,v_1' \in V_1$.
 Also, $k_2 \coloneqq \deg(v_2) = \deg(v_2')$ for all $v_2,v_2' \in V_2$ using \ref{item:bipartite-individualization-weak-4}.
 Note that $2 \leq k_2 \leq |V_1| - 2$ by \ref{item:bipartite-individualization-weak-2}.
 We also have $\lambda(v_1,v_1') = \lambda(v_1'',v_1''')$ for all $v_1,v_1',v_1'',v_1''' \in V_1$ such that $v_1 \neq v_1'$ and $v_1'' \neq v_1'''$.
 In particular, $\mu_1 \coloneqq |N(v_1) \cap N(v_1')| = |N(v_1'') \cap N(v_1''')|$ for all $v_1,v_1',v_1'',v_1''' \in V_1$ such that $v_1 \neq v_1'$ and $v_1'' \neq v_1'''$.
 Observe that $\mu_1 \geq 1$ by Condition \ref{item:bipartite-sr-individualization-2}.

 By moving to the complement of $G$, we may also assume without loss of generality that $k_1 \leq \frac{1}{2} \cdot |V_2|$.
 Also, $k_1 \geq 2$ using Conditions \ref{item:bipartite-individualization-weak-2} and \ref{item:bipartite-individualization-weak-3}.
 Since $|V_1| \cdot k_1 = |V_2| \cdot k_2$, we also get that $k_2 \leq \frac{1}{2} \cdot |V_1|$.
 Moreover, $k_1 \cdot k_2 \geq |V_1| - 1$ using \ref{item:bipartite-individualization-weak-3} again.

 \begin{claim}
  $k_1^2 > \mu_1 |V_2|$.
 \end{claim}
 \begin{claimproof}
  By symmetry, it suffices to prove the first statement.
  We define a hypergraph $\CH = (V_2,\CE)$ where
  \[\CE \coloneqq \{N(v_1) \mid v_1 \in V_1\}.\]
  Then $\CH$ is a regular and $k_1$-uniform hypergraph such that $|E \cap E'| \geq \mu_1$ for all $E,E' \in \CE$.
  So $k_1^2 > \mu_1 |V_2|$ by Lemma \ref{la:hypergraph-bounds}.
 \end{claimproof}

 Since $k_1 \leq \frac{1}{2}|V_2|$, we get that $\frac{1}{2}k_1 > \mu_1$.
 Let $v_1,v_1' \in V_1$.
 Then $|N(v_1) \triangle N(v_1')| = 2(k_1 - \mu_1) > k_1$.
 So, by Lemma \ref{la:distinguishing-set-from-fractional-cover}, there is a set $U \subseteq V(G)$ such that $|U| \leq \frac{n}{k_1}(1 + 2\log n)$ and $\chi_{2,U}(v_1) \neq \chi_{2,U}(v_1')$ for all distinct $v_1,v_1' \in V_1$.
 Using Claim \ref{claim:individualize-one-side} we get that $\chi_{2,U}$ is discrete.

 To complete the proof, it remains to bound $n/k_1$.
 We have
 \begin{align*}
  k_1 = \frac{|V_2|}{|V_1|} \cdot k_2 \geq \frac{n^{3/4}}{|V_1|} \cdot k_2 \geq |V_1|^{-1/4} \cdot k_2.
 \end{align*}
 This implies that
 \begin{align*}
  k_1^2 \geq |V_1|^{-1/4} \cdot k_1 \cdot k_2
        \geq |V_1|^{-1/4} \cdot (|V_1| - 1)
        \geq \frac{1}{2} \cdot |V_1|^{3/4}
        =    \frac{1}{2} \cdot \left(|V_1|^{4/3}\right)^{9/16}
        \geq \frac{1}{2} \cdot n^{9/16}.
 \end{align*}
 It follows that
 \[k_1 \geq \frac{1}{\sqrt{2}} \cdot n^{9/32}\]
 which means that
 \[\frac{n}{k_1} \leq \sqrt{2} \cdot n^{23/32} \leq \sqrt{2} \cdot n^{3/4}.\]
 Overall, we get that
 \[|U| \leq \frac{n}{k_1}(1 + 2\log n) \leq 6 n^{3/4} \log n.\qedhere\]
\end{proof}

The next lemma allows us to cover the case $|\im(\chi)| \in \{2,3\}$.

\begin{lemma}
 \label{la:split-with-2-or-3-vertex-colors}
 Let $(G,\chi)$ be a nice and twin-free graph.
 Also suppose $\xi_{G,\chi}(c) \leq 2$ for all colors $c \in \im(\chi)$, and $|\im(\chi)| \in \{2,3\}$.

 Then there is some $U \subseteq V(G)$ such that $1 \leq |U| \leq 2$ and $|\im(\chi_{2,U})| \geq |\im(\chi)| + 4 \cdot |U|$, or there is a set $U \subseteq V(G)$ of size
 \[|U| \leq 6n^{3/4}\log n\]
 such that $\chi_{2,U}$ is discrete.
\end{lemma}

\begin{proof}
 We start by proving the following claim.
 \begin{claim}
  \label{claim:individualize-one-color-class}
  Let $U \subseteq V(G)$ such that there is some color $c \in \im(\chi)$ such that $\chi_{2,U}(v) \neq \chi_{2,U}(v')$ for all distinct $v,v' \in \chi^{-1}(c)$.
  Then $\chi_{2,U}$ is discrete.
 \end{claim}
 \begin{claimproof}
  Let $\lambda \coloneqq \WL{2}{G,\chi}$.
  First suppose that $|\im(\chi)| = 2$ and let $d$ denote the second color in $\im(\chi)$ apart from $c$.
  Let $V \coloneqq \chi^{-1}(c)$ and $W \coloneqq \chi^{-1}(d)$.
  To show the claim, it suffices to argue that $N(w) \cap V \neq N(w') \cap V$ for all distinct $w,w' \in W$.
  Suppose towards a contradiction that there are distinct $w,w' \in W$ such that $N(w) \cap V = N(w') \cap V$.
  Consider the corresponding equivalence relation $\sim$ defined on $W$ via $w \sim w'$ if $N(w) \cap V = N(w') \cap V$.
  Let $A_1,\dots,A_\ell$ denote the equivalence classes.
  Note that $\ell > 1$ since $G$ is connected and flipped.
  Since pairs of equivalent vertices receive a different color under $\lambda$ than non-equivalent vertices and $|\lambda(W,W)| \leq 3$ (because $\mu_{G,\chi}(d,d) \leq 2$), we conclude that $G[W]$ is either empty or a disjoint union of cliques corresponding to $A_1,\dots,A_\ell$.
  In both cases, distinct equivalent vertices are twins in $G$ which contradicts $G$ being twin-free.

  The proof for the case $|\im(\chi)| = 3$ is similar.
 \end{claimproof}

 First suppose that $|\im(\chi)| = 3$.
 Assume that $\im(\chi) = \{c_1,c_2,c_3\}$ and $\mu_{G,\chi}(c_1,c_2) = \mu_{G,\chi}(c_2,c_3) = 2$.
 We say a pair $(c_i,c_j)$ of distinct colors induces a matching if $G$ induces a matching between $\chi^{-1}(c_i)$ and $\chi^{-1}(c_j)$.
 If $(c_1,c_2)$ or $(c_2,c_3)$ does not induce a matching, then the lemma follows from Lemma \ref{la:bipartite-individualization-weak} and Claim \ref{claim:individualize-one-color-class}.

 So suppose that both $(c_1,c_2)$ and $(c_2,c_3)$ induce a matching.
 If $\mu_{G,\chi}(c_1,c_3) = 2$ then it cannot induced a matching.
 So the lemma again follows from Lemma \ref{la:bipartite-individualization-weak} and Claim \ref{claim:individualize-one-color-class}.

 Since $\xi_{G,\chi}(c) \leq 2$ for all $c \in \im(\chi)$, we get that $\mu_{G,\chi}(c_2,c_2) = 1$.
 Also, it holds that $\mu_{G,\chi}(c_1,c_1),\mu_{G,\chi}(c_3,c_3) \in \{1,2\}$.
 Since $(G,\chi)$ is connected and flipped, it follows that there is some $i \in \{1,3\}$ such that $G[\chi^{-1}(c_i)]$ is connected and strongly regular.
 In this case, the lemma follows from Theorem \ref{thm:uniprimitive-individualization} and Claim \ref{claim:individualize-one-color-class}.

 So it remains to cover the case $|\im(\chi)| = 2$.
 Suppose that $\im(\chi) = \{c_1,c_2\}$.
 We have $\mu_{G,\chi}(c_1,c_1),\mu_{G,\chi}(c_2,c_2) \in \{1,2\}$.
 First suppose that $(c_1,c_2)$ induces a matching.
 Then the graph $G[\chi^{-1}(c_1)]$ is connected and strongly regular, because $(G,\chi)$ is connected and flipped.
 As before, the lemma follows from Theorem \ref{thm:uniprimitive-individualization} and Claim \ref{claim:individualize-one-color-class}.

 So we may assume that $(c_1,c_2)$ does not induce a matching.
 If we have $\mu_{G,\chi}(c_1,c_1) = 1$ or $\mu_{G,\chi}(c_2,c_2) = 1$, the lemma follows from Lemma \ref{la:bipartite-individualization-weak} and Claim \ref{claim:individualize-one-color-class}.
 So suppose that additionally $\mu_{G,\chi}(c_1,c_1) = \mu_{G,\chi}(c_2,c_2) = 2$.

 Let us denote $V_1 \coloneqq \chi^{-1}(c_1)$ and $V_2 \coloneqq \chi^{-1}(c_2)$.
 Also let $\lambda \coloneqq \WL{2}{G,\chi}$.
 Let $p_1,q_1$ be the non-diagonal colors in $\lambda(V_1,V_1)$ and $p_2,q_2$ be the non-diagonal colors in $\lambda(V_2,V_2)$.
 If $F_{p_1}$ and $F_{q_1}$ are both connected, then the lemma follows from Theorem \ref{thm:uniprimitive-individualization} and Claim \ref{claim:individualize-one-color-class}.
 The same applies if $F_{p_2}$ and $F_{q_2}$ are both connected.

 So suppose without loss of generality that $F_{p_1}$ and $F_{p_2}$ are not connected.
 In this case, both of these graphs are disjoint union of cliques of the same size.
 Let $a_1$ denote the clique size in $F_{p_1}$ and $a_2$ the clique size in $F_{p_2}$.
 Finally, suppose that $\lambda(V_1,V_2) = \{p,q\}$.

 First suppose that $a_1 = 2$ or $a_2 = 2$.
 By symmetry, we may assume without loss of generality that $a_1 = 2$.
 Let $n_2 \coloneqq |V_2|$.

 \begin{claim}
  $|N(v) \triangle N(v')| \geq n_2/2$ for all distinct $v,v' \in V_1$.
 \end{claim}
 \begin{claimproof}
  Suppose that $A_1,\dots,A_{\ell_1}$ denote the vertex sets of the connected components of $F_{p_1}$, and also assume that $A_i = \{v_i,w_i\}$ for all $i \in [\ell_1]$.

  We first argue that $N(v_i) \cap N(w_i) \cap V_2 = \emptyset$ and $V_2 \subseteq N(v_i) \cup N(w_i)$ for all $i \in [\ell_1]$.
  Indeed, there is some $x \in V_2$ such that $\lambda(v_i,x) \neq \lambda(w_i,x)$, because $v_i$ and $w_i$ are no twins.
  However, this implies that $\lambda(v_i,x) \neq \lambda(w_i,x)$ for all $x \in V_2$.
  Since $|\lambda(V_1,V_2)| = 2$, it follows that $v_ix \in E(G)$ if and only if $w_ix \notin E(G)$ for all $x \in V_2$.

  Since $|N(v_i) \cap V_2| = |N(w_i) \cap V_2|$, we also get that $|N(v_i) \cap V_2| = |N(w_i) \cap V_2| = n_2/2$ for all $i \in [\ell_1]$.
  In particular, $|N(v_i) \triangle N(w_i)| \geq n_2/2$ for all $i \in [\ell_1]$.

  Now, consider two distinct $i,j \in [\ell_1]$.
  Since
  \[\lambda(v_i,v_j) = \lambda(w_i,v_j) = \lambda(v_i,w_j) = \lambda(w_i,w_j)\]
  we conclude that
  \begin{align*}
      &|N(v_i) \cap N(v_j) \cap V_2| = |N(w_i) \cap N(v_j) \cap V_2|\\
   =~ &|N(v_i) \cap N(w_j) \cap V_2| = |N(w_i) \cap N(w_j) \cap V_2|.
  \end{align*}
  Since
  \begin{align*}
      &|N(v_i) \cap N(v_j) \cap V_2| + |N(w_i) \cap N(v_j) \cap V_2|\\
   +~ &|N(v_i) \cap N(w_j) \cap V_2| + |N(w_i) \cap N(w_j) \cap V_2| = n_2
  \end{align*}
  it follows that
  \begin{align*}
      &|N(v_i) \cap N(v_j) \cap V_2| = |N(w_i) \cap N(v_j) \cap V_2|\\
   =~ &|N(v_i) \cap N(w_j) \cap V_2| = |N(w_i) \cap N(w_j) \cap V_2| = \frac{n_2}{4}.
  \end{align*}
  This implies the claim.
 \end{claimproof}

 Combining the last claim and Lemma \ref{la:distinguishing-set-from-fractional-cover} we obtain a set $U' \subseteq V(G)$ of size $|U'| \leq 2 \frac{n}{n_2} (1 + 2 \log n)$ such that $\chi_{2,U'}(v) \neq \chi_{2,U'}(v')$ for all $v,v' \in V_1$.
 Then $\chi_{2,U'}$ is discrete by Claim \ref{claim:individualize-one-color-class}.
 If $n_2 \geq \sqrt{n}$ then $|U'| \leq 6 \cdot \sqrt{n} \log n$ and we set $U \coloneqq U'$.
 Otherwise, $n_2 \leq \sqrt{n}$ and we simply set $U \coloneqq V_2$.
 Note that $\chi_{2,U}$ is discrete by Claim \ref{claim:individualize-one-color-class}.

 It remains to consider the case that $a_1 \geq 3$ and $a_2 \geq 3$.
 Let $C_{1,1},\dots,C_{1,\ell_1}$ denote the vertex sets of the connected components of $F_{p_1}$.
 Similarly, let $C_{2,1},\dots,C_{2,\ell_2}$ denote the vertex sets of the connected components of $F_{p_2}$.

 \begin{claim}
  \label{claim:colors-incident-to-vertex}
  Let $i_1 \in [\ell_1]$, $i_2 \in [\ell_2]$, $v_1 \in V_1$ and $v_2 \in V_2$.
  Then $\lambda(v_1,C_{2,i_2}) = \lambda(C_{1,i_1},v_2) = \{p,q\}$.
 \end{claim}

 \begin{claimproof}
  Suppose towards a contradiction that $|\lambda(v_1,C_{2,i_2})| = 1$, i.e., $\lambda(v_1,w_2) = \lambda(v_1,w_2')$ for all $w_2,w_2' \in C_{2,i_2}$.
  Without loss of generality suppose that $\lambda(v_1,w_2) = p$ for all $w_2 \in C_{2,i_2}$.

  Let $w_2,w_2' \in C_{2,i_2}$ be distinct elements.
  Since $G$ is twin-free, we conclude that there is some $u_1 \in V_1$ such that
  \[\{\lambda(u_1,w_2),\lambda(u_1,w_2')\} = \{p,q\}.\]
  We may assume without loss of generality that $\lambda(u_1,w_2) = p$ and $\lambda(u_1,w_2') = q$.
  In particular, there is some vertex $x_2 \in V_2$ such that $\lambda(u_1,w_2) = p$,  $\lambda(u_1,x_2) = q$ and $\lambda(w_2,x_2) = p_2$.

  By the properties of $2$-WL, there also is some $y_2 \in V_2$ such that $\lambda(v_1,w_2) = p$,  $\lambda(v_1,y_2) = q$ and $\lambda(w_2,y_2) = p_2$.
  This is a contradiction since $y_2$ needs to be contained in $C_{2,i_2}$.
 \end{claimproof}

 Now, we pick two arbitrary elements $u_1 \in V_1$ and $u_2 \in V_2$ such that $u_1u_2 \notin E(G)$.
 We set $U \coloneqq \{u_1,u_2\}$.
 We claim that $|\im(\chi_{2,U})| \geq 10$.

 Let $i_1 \in [\ell_1]$, $i_2 \in [\ell_2]$ such that $u_1 \in C_{1,i_1}$ and $u_2 \in C_{2,i_2}$.
 Clearly, the color sets $\chi_{2,U}(C_{1,i_1})$, $\chi_{2,U}(C_{2,i_2})$, $\chi_{2,U}(V_1 \setminus C_{1,i_1})$ and $\chi_{2,U}(V_2 \setminus C_{2,i_2})$ are pairwise disjoint.
 We show that the first two sets contain at least three colors and the last two sets contain at least two colors.
 The second part follows immediately from Claim \ref{claim:colors-incident-to-vertex}.

 So consider the set $C_{1,i_1}$.
 Clearly, vertices in $N(u_2) \cap C_{1,i_1}$ are assigned a different color than vertices from $C_{1,i_1} \setminus N(u_2)$.
 Also,
 \[|N(u_2) \cap C_{1,i_1}| \leq \frac{1}{2}|C_{1,i_1}| \leq |C_{1,i_1}| - 2.\]
 The first inequality holds since $u_2$ has the same number of neighbors in every set $C_{1,j_1}$ for $j_1 \in [\ell_1]$, and $G$ is flipped.
 The second inequality holds because $|C_{1,i_1}| = a_1 \geq 3$.
 Since $u_1 \notin N(u_2)$ it follows that $|\chi_{2,U}(C_{1,i_1})| \geq 3$.
 We obtain $|\chi_{2,U}(C_{2,i_2})| \geq 3$ analogously.
\end{proof}

Finally, it remains to the cover the case $|\im(\chi)| = 1$.
For $m \geq 1$ define the graph $L_{2,m}$ with vertex set $L_{2,m} \coloneqq [2] \times [m]$ and edge set
\[E(L_{2,m}) \coloneqq \{(i,j)(i',j') \mid i = i' \vee j = j'\}.\]

\begin{lemma}
 \label{la:split-with-1-vertex-color}
 Let $(G,\chi)$ be a nice and twin-free graph.
 Also suppose that $|\im(\chi)| = 1$ and $\xi_{G,\chi}(c_0) \leq 2$ for the unique color $c_0 \in \im(\chi)$.

 Then one of the following holds:
 \begin{enumerate}[label=(\roman*)]
  \item\label{item:split-with-1-vertex-color-1} there is a set $U \subseteq V(G)$ of size $1 \leq |U| \leq 2$ such that $|\im(\chi_{2,U})| \geq 1 + 4 \cdot |U|$,
  \item\label{item:split-with-1-vertex-color-2} there is a set $U \subseteq V(G)$ of size
   \[|U| < \max(6\sqrt{n}\log n,24 \log n)\]
   such that $\chi_{2,U}$ is discrete, or
  \item\label{item:split-with-1-vertex-color-3} $G$ is isomorphic to $L_{2,m}$ for some $m \geq 3$.
 \end{enumerate}
\end{lemma}

\begin{proof}
 Let $\lambda \coloneqq \WL{2}{G,\chi}$.
 Note that $c_0 \coloneqq \lambda(v,v) = \lambda(w,w)$ for all $v,w \in V(G)$.
 Suppose $\im(\lambda) = \{c_0,c_1,\dots,c_\xi\}$ and let $F_i \coloneqq F_{c_i}$ denote the constituent graph of color $c_i$, $i \in [\xi]$.
 Observe that $\xi = \mu_{G,\chi}(c_0,c_0)$ and thus, $\xi \leq \xi_{G,\chi}(c_0) + 1 \leq 3$
 If $F_1,\dots,F_\xi$ are connected, then there is a set $U \subseteq V(G)$ of size
 \[|U| < 4\sqrt{n}\log n\]
 such that $\chi_{2,U}$ is discrete by Theorem \ref{thm:uniprimitive-individualization}.

 So we may assume that there is some $i \in [\xi]$ such that $F_i$ is not connected.
 In particular, since $(G,\chi)$ is connected and flipped, this implies that $\xi = 3$.
 Additionally, we also get that $\lambda(v,w) = \lambda(w,v)$ for all $v,w \in V(G)$.

 Without loss of generality, assume that $F_1$ is not connected.
 Since $(G,\chi)$ is connected and flipped and $\xi = 3$, we conclude that $F_1$ is a disjoint union of complete graphs.
 Let $A_1,\dots,A_\ell$ denote the partition of $V(G)$ into the vertex sets of the connected components of $F_1$.
 Note that $a \coloneqq |A_i| = |A_j|$ for all $i,j \in [\ell]$.

 We first claim that $|\lambda(A_i \times A_j)| = 2$ for all distinct $i,j \in [\ell]$.
 Indeed, if there were $i,j \in [\ell]$ such that $|\lambda(A_i \times A_j)| = 1$, then $|\lambda(A_i \times A_j)| = 1$ for all distinct $i,j \in [\ell]$, and any pair of distinct vertices $v_1,v_2 \in A_1$ are twins.

 So we may assume that $|\lambda(A_i \times A_j)| = 2$ for all distinct $i,j \in [\ell]$.
 In particular, there are $d_2,d_3 \geq 1$ such that
 \[|N_{F_2}(v_i) \cap A_j| = d_2\]
 and
 \[|N_{F_3}(v_i) \cap A_j| = d_3\]
 for all distinct $i,j \in [\ell]$, and all $v_i \in A_i$.
 Note that $d_2 + d_3 = a$.
 Without loss of generality assume that $d_2 \leq d_3$.

 \medskip

 Suppose that $\ell \geq 3$ and $a \geq 3$.
 We say a pair $(v_1,v_2)$, where $v_1 \in A_1$ and $v_2 \in A_2$, is \emph{good} if
 \begin{enumerate}[label=(\roman*)]
  \item\label{item:good-1} $N_{F_2}(v_1) \cap A_2 \neq \{v_2\}$ and $N_{F_3}(v_1) \cap A_2 \neq \{v_2\}$,
  \item\label{item:good-2} $N_{F_2}(v_2) \cap A_1 \neq \{v_1\}$ and $N_{F_3}(v_2) \cap A_1 \neq \{v_1\}$, and
  \item\label{item:good-3} $\{N_{F_2}(v_1) \setminus (A_1\cup A_2), N_{F_3}(v_1) \setminus (A_1\cup A_2)\} \neq \{N_{F_2}(v_2) \setminus (A_1\cup A_2), N_{F_3}(v_2) \setminus (A_1\cup A_2)\}$
 \end{enumerate}

 \begin{claim}
  Let $(v_1,v_2)$ be a good pair and define $U \coloneqq \{v_1,v_2\}$.
  Then $|\im(\chi_{2,U})| \geq 9$.
 \end{claim}

 \begin{claimproof}
  Clearly, the partition into color classes of $\im(\chi_{2,U})$ refines the partition $\{A_1,A_2,V(G) \setminus (A_1 \cup A_2)\}$ since $v_1 \in A_1$ and $v_2 \in A_2$.
  We have $|\chi_{2,U}(A_1)| \geq 3$ and $|\chi_{2,U}(A_2)| \geq 3$ by Conditions \ref{item:good-1} and \ref{item:good-2}.
  Finally, Condition \ref{item:good-3} ensures that $|\chi_{2,U}(V(G) \setminus (A_1 \cup A_2))| \geq 3$.
 \end{claimproof}

 Hence, it suffices to argue that there is a good pair $(v_1,v_2)$.
 First suppose that $d_2 = 1$.
 Let $v_1$ be an arbitrary vertex.
 Then there are at least two vertices $v_2,v_2' \in A_2$ such that $v_2,v_2' \notin N_{F_2}(v_1)$.
 Let $v_3$ be the unique element in $A_3 \cap N_{F_2}(v_1)$.
 Then $v_2 \notin N_{F_2}(v_3)$ or $v_2' \notin N_{F_2}(v_3)$.
 It follows that $(v_1,v_2)$ is good or $(v_1,v_2')$ is good.

 So assume that $2 \leq d_2 \leq d_3$.
 This means Conditions \ref{item:good-1} and \ref{item:good-2} are satisfied for all $v_1 \in A_1$ and $v_2 \in A_2$.
 We say two vertices $v_2,v_2' \in A_2$ are
 \begin{itemize}
  \item \emph{related} if $N_{F_2}(v_2) \setminus (A_1 \cup A_2) = N_{F_2}(v_2') \setminus (A_1 \cup A_2)$ and $N_{F_3}(v_2) \setminus (A_1 \cup A_2) = N_{F_3}(v_2') \setminus (A_1 \cup A_2)$, and
  \item \emph{flip-related} if $N_{F_2}(v_2) \setminus (A_1 \cup A_2) = N_{F_3}(v_2') \setminus (A_1 \cup A_2)$ and $N_{F_3}(v_2) \setminus (A_1 \cup A_2) = N_{F_2}(v_2') \setminus (A_1 \cup A_2)$,
 \end{itemize}

 \begin{claim}
  Suppose $|V(G)| \geq 13$.
  Then there are distinct $v_2,v_2' \in A_2$ that are neither related nor flip-related.
 \end{claim}

 \begin{claimproof}
  Suppose towards a contradiction that all distinct $v_2,v_2' \in A_2$ are related or flip-related.

  First suppose that $a \geq 5$.
  Let $v_1 \in A_1$ and $v_3 \in A_3$ be arbitrary elements.
  Then there are $i,j \in \{2,3\}$ such that
  \[|A_2 \cap N_{F_i}(v_1) \cap N_{F_j}(v_3)| \geq 2.\]
  Let $v_2,v_2' \in A_2 \cap N_{F_i}(v_1) \cap N_{F_j}(v_3)$ be distinct elements.
  Since $v_2,v_2' \in N_{F_j}(v_3)$ we get that $v_2,v_2'$ are not flip-related.
  So they are related.
  It follows that
  \begin{align*}
   |\{x \in V(G) \setminus A_2 \mid \lambda(v_2,x) = \lambda(v_2',x)\}| \geq |V(G) \setminus (A_1 \cup A_2)| + 1 = a \cdot (\ell - 2) + 1 \geq a + 1.
  \end{align*}
  Now let $v_2'' \in A_2 \setminus N_{F_j}(v_3)$.
  Then $v_2,v_2''$ are not related, so they are flip-related.
  Hence,
  \[|\{x \in V(G) \setminus A_2 \mid \lambda(v_2,x) = \lambda(v_2'',x)\}| \leq |A_1| \leq a.\]
  But this is a contradiction since $\lambda(v_2,v_2') = \lambda(v_2,v_2'') = c_1$.

  So it remains to consider the case $a \leq 4$.
  Then $\ell \geq 4$ since $|V(G)| = a \cdot \ell \geq 13$ by assumption.
  Let $v_3 \in A_3$ be an arbitrary element.
  We pick distinct elements $v_2,v_2' \in N_{F_2}(v_3)$ and $v_2'' \in A_2 \setminus N_{F_2}(v_3)$.
  Then $v_2,v_2'$ are not flip-related and $v_2,v_2''$ are not related.
  So $v_2,v_2'$ are related and $v_2,v_2''$ are flip-related.
  Similar to the analysis above, it follows that
  \begin{align*}
   |\{x \in V(G) \setminus A_2 \mid \lambda(v_2,x) = \lambda(v_2',x)\}| \geq |V(G) \setminus (A_1 \cup A_2)| = a \cdot (\ell - 2) \geq 2a
  \end{align*}
  and
  \[|\{x \in V(G) \setminus A_2 \mid \lambda(v_2,x) = \lambda(v_2'',x)\}| \leq |A_1| \leq a.\]
  Again, this is a contradiction since $\lambda(v_2,v_2') = \lambda(v_2,v_2'') = c_1$.
 \end{claimproof}

 If $|V(G)| \leq 12$ then we set $U \coloneqq V(G) \setminus \{v\}$ for some arbitrary $v \in V(G)$.
 Clearly, $\chi_{2,U}$ is discrete and $|U| < 4 \sqrt{n} \log n$.
 Otherwise, there are distinct $v_2,v_2' \in A_2$ that are neither related nor flip-related.
 Then $(v_1,v_2)$ is good or $(v_1,v_2')$ is good for every $v_1 \in A_1$.
 In particular, there is a good pair.
 This completes the case that $\ell \geq 3$ and $a \geq 3$.

 \medskip

 Next, we consider the case $\ell = 2$.
 If $d_2 = 1$, then it is easy to see that $G$ is isomorphic to $L_{2,a}$ for some $a \geq 3$.
 So suppose that $2 \leq d_2 \leq d_3$.

 We also have $\lambda(v_1,v_1') = \lambda(v_1'',v_1''')$ for all $v_1,v_1',v_1'',v_1''' \in A_1$ such that $v_1 \neq v_1'$ and $v_1'' \neq v_1'''$.
 This implies that $\mu \coloneqq |N_{F_2}(v_1) \cap N_{F_2}(v_1')| = |N_{F_2}(v_1'') \cap N_{F_2}(v_1''')|$ for all $v_1,v_1',v_1'',v_1''' \in V_1$ such that $v_1 \neq v_1'$ and $v_1'' \neq v_1'''$.
 Moreover, $\lambda(v_1,v_1') = \lambda(v_2,v_2')$ for all $v_1,v_1'\in A_1$, $v_2,v_2' \in A_2$ such that $v_1 \neq v_1'$ and $v_2 \neq v_2'$.
 So $\mu \coloneqq |N_{F_2}(v_2) \cap N_{F_2}(v_2')|$ for all distinct $v_2,v_2' \in A_2$.

 Note that $2 \leq d_2 \leq a/2$, which also implies that $\mu \geq 1$.
 Moreover, $d_2^2 \geq a - 1$.

 \begin{claim}
  $d_2^2 > \mu \cdot a$.
 \end{claim}
 \begin{claimproof}
  We define a hypergraph $\CH = (A_2,\CE)$ where
  \[\CE \coloneqq \{N_{F_2}(v_1) \mid v_1 \in A_1\}.\]
  Then $\CH$ is a regular and $d_2$-uniform hypergraph such that $|E \cap E'| \geq \mu$ for all $E,E' \in \CE$.
  So $d_2^2 > \mu \cdot a$ by Lemma \ref{la:hypergraph-bounds}.
 \end{claimproof}

 Since $d_2 \leq a/2$, we get that $\frac{1}{2}d_2 > \mu$.

 Let $v_1,v_1' \in A_1$.
 Then $|N(v_1) \triangle N(v_1')| \geq |N_{F_2}(v_1) \triangle N_{F_2}(v_1')| = 2(d_2 - \mu) > d_2$.
 So, by Lemma \ref{la:distinguishing-set-from-fractional-cover}, there is a set $U \subseteq V(G)$ such that $|U| \leq \frac{n}{d_2}(1 + 2\log n)$ and $\chi_{2,U}(v_1) \neq \chi_{2,U}(v_1')$ for all distinct $v_1,v_1' \in A_1$.
 This implies that $\chi_{2,U}$ is discrete.
 Also, $d_2^2 \geq a - 1 \geq n/2 - 1 \geq n/4$ which implies that
 \[\frac{n}{d_2}(1 + 2\log n) \leq 2\sqrt{n}(1 + 2\log n) \leq 6\sqrt{n}\log n.\]
 This completes the case $\ell = 2$.

 \medskip

 Finally, suppose that $a = 2$.
 Observe that $d_2 + d_3 + 2 = n$ and recall that $d_2 \leq d_3$.

 \begin{claim}
  \label{claim:color-between-2-blocks}
  Let $i \in [\ell]$ and $v \in V(G) \setminus A_i$.
  Then $\lambda(v,A_i) = \{c_2,c_3\}$
 \end{claim}

 \begin{claimproof}
  Suppose $A_i = \{w,w'\}$.
  Since $G$ is twin-free, there is some $v' \in V(G) \setminus A_i$ such that $\lambda(v',A_i)  = \{c_2,c_3\}$.
  It follows that for every $x,y \in V(G)$ such that $\lambda(x,y) = c_2$ there is some $z \in V(G)$ such that $\lambda(x,z) = c_1$ and $\lambda(y,z) = c_3$.
  Similarly, for every $x,y \in V(G)$ such that $\lambda(x,y) = c_3$ there is some $z \in V(G)$ such that $\lambda(x,z) = c_1$ and $\lambda(y,z) = c_2$.
  This implies the claim.
 \end{claimproof}

 The last claim implies that $d_2 = d_3 = \ell - 1 = \frac{n-2}{2}$.
 Now, first assume that $F_2$ or $F_3$ is disconnected.
 Without loss of generality, suppose that $F_2$ is disconnected.
 Since all connected components of $F_2$ have the same size, it follows that $F_2$ consists of two cliques of size $n/2$.
 It follows that $G$ is isomorphic to $L_{2,\ell}$, where $\ell \geq 3$.

 So we may assume that $F_2$ and $F_3$ are both connected.
 By the properties of $2$-WL, there are integers $\mu_2,\mu_3 \geq 1$ such that
 \[|N_{F_2}(v) \cap N_{F_2}(w)| = \mu_2\]
 for all $v,w \in V(G)$ such that $\lambda(v,w) = c_2$, and
 \[|N_{F_2}(v) \cap N_{F_2}(w)| = \mu_3\]
 for all $v,w \in V(G)$ such that $\lambda(v,w) = c_3$.

 \begin{claim}
  \label{claim:compare-mu}
  $d_2 - \mu_3 \leq 2\cdot(d_2 - \mu_2)$ and $d_2 - \mu_2 \leq 2\cdot(d_2 - \mu_3)$.
 \end{claim}
 \begin{claimproof}
  Pick $v,w \in V(G)$ such that $\lambda(v,w) = c_3$.
  Also, let $z \in N_{F_2}(v) \cap N_{F_2}(w)$.
  Then
  \[N_{F_2}(v) \setminus N_{F_2}(w) \subseteq \Big(N_{F_2}(v) \setminus N_{F_2}(z)\Big) \cup \Big(N_{F_2}(z) \setminus N_{F_2}(w)\Big).\]
  It follows that
  \[d_2 - \mu_3 = |N_{F_2}(v) \setminus N_{F_2}(w)| \leq 2 \cdot (d_2 - \mu_2).\]

  For the second statement, let $v,w \in V(G)$ such that $\lambda(v,w) = c_2$.
  Note that $N_{F_3}(v) \cap N_{F_3}(w) \neq \emptyset$.
  Pick some element $z \in N_{F_3}(v) \cap N_{F_3}(w)$.
  As before
  \[N_{F_2}(v) \setminus N_{F_2}(w) \subseteq \Big(N_{F_2}(v) \setminus N_{F_2}(z)\Big) \cup \Big(N_{F_2}(z) \setminus N_{F_2}(w)\Big).\]
  It follows that
  \[d_2 - \mu_2 = |N_{F_2}(v) \setminus N_{F_2}(w)| \leq 2 \cdot (d_2 - \mu_3).\]
 \end{claimproof}

 Next, let $\mu \coloneqq \min(\mu_2,\mu_3)$.

 \begin{claim}
  \label{claim:bound-mu}
  $\mu \leq \frac{1}{2} d_2$.
 \end{claim}
 \begin{claimproof}
  Suppose $A_1 = \{v_1,w_1\}$ and pick $v_2 \in A_2$ such that
  \[\lambda(v_1,v_2) = c_2 \quad\text{ and }\quad \lambda(w_1,v_2) = c_3.\]
  Assume towards a contradiction that $\mu > \frac{1}{1} d_2$.
  Then
  \[|N_{F_2}(v_1) \cap N_{F_2}(v_2)| > \frac{1}{2} d_2\]
  and
  \[|N_{F_2}(w_1) \cap N_{F_2}(v_2)| > \frac{1}{2} d_2.\]
  Since $|N_{F_2}(v_2)| = d_2$, it follows that
  \[N_{F_2}(v_1) \cap N_{F_2}(w_1) \neq \emptyset\]
  which contradicts Claim \ref{claim:color-between-2-blocks}.
 \end{claimproof}

 Combining Claims \ref{claim:compare-mu} and \ref{claim:bound-mu} we get that $\mu_2 \leq \frac{3}{4} d_2$ and $\mu_3 \leq \frac{3}{4} d_2$.

 \begin{claim}
  $|N_{F_2}(v) \triangle N_{F_2}(w)| \geq \frac{1}{2}d_2$ for all distinct $v,w \in V(G)$.
 \end{claim}
 \begin{claimproof}
  If $\lambda(v,w) = c_1$, then $N_{F_2}(v) \cap N_{F_2}(w) = \emptyset$ by Claim \ref{claim:color-between-2-blocks}.
  So $|N_{F_2}(v) \triangle N_{F_2}(w)| = 2d_2$.

  If $\lambda(v,w) = c_2$, then $|N_{F_2}(v) \cap N_{F_2}(w)| = \mu_2 \leq \frac{3}{4} d_2$.
  It follows that $|N_{F_2}(v) \triangle N_{F_2}(w)| = 2(d_2 - \mu_2) \geq \frac{1}{2}d_2$.

  Similarly, if $\lambda(v,w) = c_3$, we get $|N_{F_2}(v) \triangle N_{F_2}(w)| = 2(d_2 - \mu_3) \geq \frac{1}{2}d_2$.
 \end{claimproof}

 Applying Lemma \ref{la:distinguishing-set-from-fractional-cover}, we obtain a set $U \subseteq V(G)$ such that $|U| \leq \frac{2n}{d_2}(1 + 2\log n)$ and $\chi_{2,U}$ is discrete.
 We have
 \[\frac{2n}{d_2}(1 + 2\log n) = \frac{4n}{n-2}(1 + 2\log n) \leq 24 \log n.\qedhere\]
\end{proof}

With this, we are finally ready to give the proof of Theorem \ref{thm:wl-depth-general}.

\begin{proof}[Proof of Theorem \ref{thm:wl-depth-general}]
 We define
 \[f(n) \coloneqq \max\big(6 n^{3/4}\log n, 24 \log n\big).\]
 We prove that
 \begin{equation}
  \label{eq:wl-depth-advanced}
  \WLd{2}{G,\chi} \leq \frac{|V(G)| - |\im(\chi)|}{4} + f(|V(G)|).
 \end{equation}
 for every vertex-colored graph $(G,\chi)$.

 We prove \eqref{eq:wl-depth-advanced} by induction on the tuple $(|V(G)|,|\im(\chi)|,|E(G)|)$.
 Consider the set $M \coloneqq \{(n,\ell,m) \in \ZZ_{\geq 0}^3 \mid \ell \leq n\}$ and observe that $(|V(G)|,|\im(\chi)|,|E(G)|) \in M$.
 For the induction, we define a linear order $\prec$ on $M$ via $(n,\ell,m) \prec (n',\ell',m')$ if $n < n'$, or $n = n'$ and $\ell > \ell'$, or $n = n'$ and $\ell = \ell'$ and $m < m'$.
 Note that we use the inverse order on the second component, i.e., for $\ell \neq \ell'$, we have $(n,\ell,m) \prec (n,\ell',m')$ if $\ell > \ell'$.
 Still, since $\ell \leq n$ for every $(n,\ell,m) \in M$, there are no infinite decreasing chains in $M$.

 For the base case, suppose that $|V(G)| = 1$.
 Then $\WLd{1}{G,\chi} = 0 = \frac{|V(G)| - |\im(\chi)|}{4} + f(|V(G)|)$ and the statement holds.

 For the inductive step, suppose $(G,\chi)$ is a colored graph with $|V(G)| > 1$.
 We distinguish several cases.
 \begin{itemize}[leftmargin=3ex]
  \item First suppose that $\chi$ is not stable with respect to $2$-WL, i.e., $\chi^* \prec \chi$ where $\chi^*(v) \coloneqq \WL{2}{G,\chi}(v)$ for all $v \in V(G)$.
   Observe that $|\im(\chi^*)| > |\im(\chi)|$, which implies that $(|V(G)|,|\im(\chi^*)|,|E(G)|) \prec (|V(G)|,|\im(\chi)|,|E(G)|)$.
   Hence, by the induction hypothesis, we get
   \[\WLd{2}{G,\chi^*} \leq \frac{|V(G)| - |\im(\chi^*)|}{4} + f(|V(G)|).\]
   Also, $\WLd{2}{G,\chi} \leq \WLd{2}{G,\chi^*}$ by Lemma \ref{la:wl-depth-bound}\ref{item:wl-depth-bound-1}.
   Together, we obtain that
   \begin{align*}
    \WLd{2}{G,\chi} \leq \WLd{2}{G,\chi^*} &\leq \frac{|V(G)| - |\im(\chi^*)|}{4} + f(|V(G)|)\\
                                           &\leq \frac{|V(G)| - |\im(\chi)|}{4} + f(|V(G)|).
   \end{align*}
  \item Next, suppose that $(G,\chi)$ is not flipped.
   Define $(G',\chi') \coloneqq \flip(G,\chi)$.
   Observe that $V(G') = V(G)$, $\chi' = \chi$ and $|E(G')| < |E(G)|$ hold.
   We conclude that $(|V(G')|,|\im(\chi')|,|E(G')|) \prec (|V(G)|,|\im(\chi)|,|E(G)|)$.
   Hence, by the induction hypothesis, we get
   \[\WLd{2}{G',\chi'} \leq \frac{|V(G')| - |\im(\chi')|}{4} + f(|V(G')|).\]
   Also, $\WLd{2}{G,\chi} \leq \WLd{2}{G',\chi'}$ by Lemma \ref{la:wl-depth-bound}\ref{item:wl-depth-bound-2}.
   Together, we obtain that
   \begin{align*}
    \WLd{2}{G,\chi} \leq \WLd{2}{G',\chi'} &\leq \frac{|V(G')| - |\im(\chi')|}{4} + f(|V(G')|)\\
                                           &= \frac{|V(G)| - |\im(\chi)|}{4} + f(|V(G)|).
   \end{align*}
  \item Suppose that $G$ is not connected and let $A_1,\dots,A_\ell$ be the vertex sets of the connected components of $G$.
   Observe that $|A_i| < |V(G)|$ for all $i \in [\ell]$.
   By the induction hypothesis, we get
   \[\WLd{2}{G[A_i],\chi|_{A_i}} \leq \frac{|A_i| - |\im(\chi|_{A_i})|}{4} + f(|A_i|)\]
   for all $i \in [\ell]$.
   Together with Lemma \ref{la:wl-depth-bound}\ref{item:wl-depth-bound-3}, we obtain
   \begin{align*}
    \WLd{2}{G,\chi} \leq \max_{i \in [\ell]} \WLd{2}{G[A_i],\chi|_{A_i}} \leq \max_{i \in [\ell]} \frac{|A_i| - |\im(\chi|_{A_i})|}{4} + f(|A_i|).
   \end{align*}
   Also, since $|A_i| - |\im(\chi|_{A_i})| \geq 0$, we get that
   \begin{align*}
    \max_{i \in [\ell]}(|A_i| - |\im(\chi|_{A_i})|) &\leq \sum_{i \in [\ell]} (|A_i| - |\im(\chi|_{A_i})|)\\
                                                    &=~ |V(G)| - \sum_{i \in [\ell]}|\im(\chi|_{A_i})| \leq |V(G)| - |\im(\chi)|.
   \end{align*}
   Additionally, $f$ is monotonically increasing.
   So $\WLd{2}{G,\chi} \leq \frac{|V(G)| - |\im(\chi)|}{4} + f(|V(G)|)$.
  \item So we may assume that $(G,\chi)$ is nice.
   Next, suppose $(G,\chi)$ is not twin-free and let $\pi$ be a twin-partition of $G$.
   Observe that $|\pi| < |V(G)|$, so by the induction hypothesis we get
   \[\WLd{2}{G/\pi,\chi/\pi} \leq \frac{|\pi| - |\im(\chi/\pi)|}{4} + f(|\pi|).\]
   In combination with Lemma \ref{la:remove-twins}, we conclude that
   \begin{align*}
    \WLd{2}{G,\chi} \leq \WLd{2}{G/\pi,\chi/\pi} \leq \frac{|\pi| - |\im(\chi/\pi)|}{4} + f(|\pi|).
   \end{align*}
   We have that $|\pi| - |\im(\chi/\pi)| \leq |V(G)| - |\im(\chi)|$.
   So
   \begin{align*}
    \WLd{2}{G,\chi} \leq \frac{|\pi| - |\im(\chi/\pi)|}{4} + f(|\pi|) \leq \frac{|V(G)| - |\im(\chi)|}{4} + f(|V(G)|)
   \end{align*}
   using again that $f$ is monotonically increasing.
  \item So we may suppose that $(G,\chi)$ is nice and twin-free.
   Next, suppose there is some $c \in \im(\chi)$ such that $\xi_{G,\chi}(c) \geq 3$.
   Let $u \in \chi^{-1}(c)$.
   Then $|\im(\chi_{2,u})| \geq |\im(\chi)| + 4$ by Lemma \ref{la:split-many-pair-colors}.
   So
   \begin{align*}
    \WLd{2}{G,\chi} \leq 1 + \WLd{2}{G,\chi_{2,u}}
                    &\leq 1 + \frac{|V(G)|-|\im(\chi_{2,u})|}{4} + f(|V(G)|)\\
                    &\leq 1 + \frac{|V(G)|-(|\im(\chi)| + 4)}{4} + f(|V(G)|)\\
                    &\leq \frac{|V(G)|-|\im(\chi)|}{4} + f(|V(G)|)
   \end{align*}
   where the first inequality follows from Lemma \ref{la:wl-depth-bound} and the second inequality holds by the induction hypothesis.
  \item So we may additionally assume that $\xi_{G,\chi}(c) \leq 2$ for every $c \in \im(\chi)$.
   Finally, we distinguish several cases based on $|\im(\chi)|$.
   \begin{itemize}[leftmargin=2.5ex]
    \item First suppose that  $|\im(\chi)| \geq 4$.
     By Lemma \ref{la:split-with-4-vertex-colors}, there is a set $U \subseteq V(G)$ of size
     \[|U| < 6\sqrt{|V(G)|}\log(|V(G)|) \leq f(|V(G)|)\]
     such that $\chi_{2,U}$ is discrete.
     By Lemma \ref{la:wl-depth-bound} we conclude that
     \begin{align*}
      \WLd{2}{G,\chi} &\leq |U| + \WLd{2}{G,\chi_{2,U}} \leq |U| + 0\\
                      &\leq f(|V(G)|) \leq \frac{|V(G)|-|\im(\chi)|}{4} + f(|V(G)|)
     \end{align*}
     where the second inequality holds since $\chi_{2,U}$ is discrete (see also Example \ref{exa:discrete-coloring}).
    \item Next, suppose that $|\im(\chi)| \in \{2,3\}$.
     We apply Lemma \ref{la:split-with-2-or-3-vertex-colors} and obtain a set $U \subseteq V(G)$ such that $1 \leq |U| \leq 2$ and $|\im(\chi_{2,U})| \geq |\im(\chi)| + 4 \cdot |U|$, or a set $U \subseteq V(G)$ of size
     \[|U| \leq 6|V(G)|^{3/4}\log(|V(G)|) \leq f(|V(G)|)\]
     such that $\chi_{2,U}$ is discrete.
     In the latter case, we get that $\WLd{2}{G,\chi} \leq f(|V(G)|) \leq \frac{|V(G)|-|\im(\chi)|}{4} + f(|V(G)|)$ similar to the previous case.
     In the former case, we have
     \begin{align*}
      \WLd{2}{G,\chi} &\leq |U| + \WLd{2}{G,\chi_{2,U}}\\
                      &\leq |U| + \frac{|V(G)|-|\im(\chi_{2,U})|}{4} + f(|V(G)|)\\
                      &\leq |U| + \frac{|V(G)|-(|\im(\chi)| + 4|U|)}{4} + f(|V(G)|)\\
                      &\leq \frac{|V(G)|-|\im(\chi)|}{4} + f(|V(G)|)
     \end{align*}
     where the first inequality follows from Lemma \ref{la:wl-depth-bound} and the second inequality holds by the induction hypothesis.
    \item Finally, assume that $|\im(\chi)| = 1$, in which case we use Lemma \ref{la:split-with-1-vertex-color}.
     If Option \ref{item:split-with-1-vertex-color-1} or \ref{item:split-with-1-vertex-color-2} is satisfied, we proceed as in the previous case.
     So suppose Option \ref{item:split-with-1-vertex-color-3} is satisfied, i.e., $G$ is isomorphic to $L_{2,m}$ for some $m \geq 3$.
     Then it is easy to check that $\WLd{2}{G,\chi} = 1 \leq \frac{|V(G)|-|\im(\chi)|}{4} + f(|V(G)|)$ (after individualizing an arbitrary vertex and performing $2$-WL, repeatedly flipping and splitting into connected components results in graphs with only a single vertex).\qedhere
   \end{itemize}
 \end{itemize}
\end{proof}

\section{Lower Bound on the WL Dimension}
\label{sec:lower-bound}

We provide a lower bound on the WL dimension of $n$-vertex graphs.
The proof is based on the following theorem, which follows from \cite[Lemma 5 \& Corollary 7]{DvorakN16}.

\begin{theorem}
 \label{thm:3-regular-large-tree-width}
 There exists an integer $n_0$ such that for every even $n \geq n_0$, there exists a $3$-regular graph $G$ on $n$ vertices such that
 \[\tw(G) \geq \frac{1}{24}n - 1.\]
\end{theorem}

To obtain graphs with high WL dimension from $3$-regular graphs of large tree-width, we rely on the Cai-Fürer-Immerman construction \cite{CaiFI92}.
In order to keep the number of vertices as small as possible, we use a slightly modified construction (see, e.g., \cite{Neuen23,Roberson22}).

Let $G$ be a graph and let $U \subseteq V(G)$.
For $v \in V(G)$, we define $\delta_{v,U} \coloneqq |\{v\} \cap U|$.
Also, we write $E(v) \coloneqq \{e \in E(G) \mid v \in e\}$ to denote the set edges incident to $v$.
We define the graph $\CFI(G,U)$ with vertex set
\[V(\CFI(G,U)) \coloneqq \{(v,S) \mid v \in V(G), S \subseteq E(v), |S| \equiv \delta_{v,U} \bmod 2\}\]
and edge set
\[V(\CFI(G,U)) \coloneqq \{(v,S)(u,T) \mid uv \in E(G), uv \notin S \bigtriangleup T\}.\]
Note that if $G$ is a $3$-regular graph with $n$ vertices, then we have $|V(\CFI(G,U))| = 4n$.
Indeed, for every $v \in V(G)$, the number of distinct sets $S$ that occur as a second component in vertices $(v,S)$ of $\CFI(G,U)$ is $2^3/2 = 4$, so $|V(\CFI(G,U))| = 4 |V(G)| = 4n$.

The following lemma is well known (see, e.g., \cite{CaiFI92,Roberson22}).

\begin{lemma}
 Let $G$ be a connected graph and let $U,U' \subseteq V(G)$.
 Then $\CFI(G,U) \cong \CFI(G,U')$ if and only if $|U| \equiv |U'| \bmod 2$.
\end{lemma}

We define $\CFI(G) \coloneqq \CFI(G,\emptyset)$ and $\twCFI(G) \coloneqq \CFI(G,\{u\})$ for some $u \in V(G)$.

The next lemma relates the WL dimension of $\CFI(G)$ to the tree-width of the base graph $G$.
A variant of this lemma already appears for example in \cite{DawarR07} with the underlying ideas dating back to \cite{CaiFI92}.
The concrete statement given below follows from \cite[Lemma 4.4]{Neuen23} and \cite{CaiFI92,ImmermanL90}.

\begin{lemma}
 \label{la:cfi-wl}
 Let $G$ be a connected graph of tree-width $\tw(G) > k$.
 Then $\CFI(G) \simeq_k \twCFI(G)$.
\end{lemma}

Now, we can combine Theorem \ref{thm:3-regular-large-tree-width} and Lemma \ref{la:cfi-wl} to obtain a concrete lower bound on the WL dimension in terms of the number of vertices.

\begin{theorem}
 There exists an integer $n_0$ such that for every $n \geq n_0$ that is divisible by $8$, there exist non-isomorphic $n$-vertex graphs $G_n$ and $H_n$ such that $G_n \simeq_k H_n$ for every $k < \frac{1}{96}n - 1$.
\end{theorem}

\begin{proof}
 Let $n_0'$ denote the constant from Theorem \ref{thm:3-regular-large-tree-width} and define $n_0 \coloneqq 4n_0'$.
 Let $n \geq n_0$ denote an integer that is divisible by $8$, and let $n' \coloneqq n/4$.
 Observe that $n'$ is even.
 So by Theorem \ref{thm:3-regular-large-tree-width}, there exists a $3$-regular graph $B_n$ on $n'$ vertices such that
 \[\tw(B_n) \geq \frac{1}{24}n' - 1.\]
 We define $G_n \coloneqq \CFI(B_n)$ and $H_n \coloneqq \twCFI(B_n)$.
 We have $|V(G_n)| = |V(H_n)| = 4n' = n$ as desired.
 Also, $G_n \simeq_k H_n$ for every $k < \tw(B_n)$ by Lemma \ref{la:cfi-wl}.
 So $G_n \simeq_k H_n$ for every
 \[k < \frac{1}{24}n' - 1 = \frac{1}{96}n - 1.\qedhere\]
\end{proof}

\section{Conclusion}

We have shown that the WL dimension of every $n$-vertex graph is in $(\frac{1}{4} + o(1))n$.
This implies that every $n$-vertex graph can be defined in the counting logic $\LC{}$ with $(\frac{1}{4} + o(1))n$ many variables.
Our contribution improves on the previous bound of $\frac{n+3}{2}$ on the number of variables \cite{PikhurkoVV06}.
In fact, they proved that for every two non-isomorphic $n$-vertex graphs $G$ and $H$, there is an $\FO$-formula with at most $\frac{n+3}{2}$ variables that distinguishes $G$ and $H$.
Since their bound is (essentially) tight for the logic $\FO$, our results for $\LC{}$ also yield that the ability to count does allow us to save variables when defining graphs.

To obtain the results, we have introduced the concept of the \emph{$k$-WL depth} of a graph and have proved that the $2$-WL depth of every $n$-vertex graph is at most $(\frac{1}{4} + o(1))n$.
Still, there is a significant gap towards the lower bound of $(\frac{1}{96} - o(1))n$ that can be obtained via the CFI construction.
By pushing the arguments developed in this paper, it seems possible to obtain further improvements of the upper bound, but the case analyses will become significantly more complex.

The WL depth is a concept that might raise interest in the context of graph identification beyond the scope of this work.
We have shown that the $1$-WL depth of a graph with vertex cover number $r$ is at most $\frac{2}{3}r + 1$, which implies the first non-trivial upper bound on the WL dimension in terms of $r$.
Can similar results be obtained for other graph parameters?
For example, it is not difficult to show that the $k$-WL depth of a graph $G$ is at most the tree-depth of $G$ (for every $k \geq 1$).
Are there a $k \in \mathbb{N}$ and an $\varepsilon > 0$ such that $\WLd{k}{G,\chi} \leq (1 - \varepsilon + o(1))t$ for every vertex-colored graph $(G,\chi)$ of tree-depth at most $t$?

\bibliographystyle{plainurl}
\small
\bibliography{references}

\end{document}